\documentclass[twoside]{article}
\usepackage[accepted]{aistats2014}
\usepackage{amsmath, graphicx, layout, verbatim}
\usepackage{amsthm}
\usepackage{color}
\usepackage{float}
\usepackage[colorlinks=true,urlcolor=black,citecolor=black	,linkcolor=black]{hyperref}
\usepackage{amsfonts}
\usepackage{amsbsy}
\usepackage{algorithm}
\usepackage{algpseudocode}
\usepackage{amsmath}
\usepackage{graphicx}
\usepackage{wrapfig}
\usepackage{wasysym}
\usepackage[round]{natbib}

\newcommand{\nsubG}{n_G}

\renewcommand{\comment}[1]{}

\renewcommand{\vec}[1]{\mathbf{#1}}

\def\I{{\mathbb I}}
\def\P{{\mathbb P}}
\def\L{{\mathcal L}}

\def\x{{\vec{x}}}

\def\E{{\mathbb E}}
\def\V{{\mathbb V}}

\newtheorem{thm}{Theorem}

\newtheorem{Lemma}{Lemma}
\newtheorem{Assumption}{Assumption}

\begin{document}

%

%

\twocolumn[

\aistatstitle{High-Dimensional Density Ratio Estimation with Extensions to Approximate Likelihood Computation}

\aistatsauthor{ Rafael Izbicki \And Ann B. Lee \And Chad M. Schafer }

\aistatsaddress{ Department of Statistics -- Carnegie Mellon University } ]

\begin{abstract}
  The ratio between two probability density functions is an important component of various
  tasks, including selection bias correction, novelty detection and classification.
Recently, several estimators of this ratio have been proposed. Most of these methods 
fail if the sample space is high-dimensional,
and hence require a dimension reduction step, the result of which can be a significant loss of information.
Here we propose a simple-to-implement, fully nonparametric density ratio estimator that expands 
the ratio in terms of the eigenfunctions of a kernel-based operator; these functions reflect the underlying 
geometry of the data (e.g., submanifold structure), often leading to better estimates without an explicit dimension reduction step.
We show how our general framework can be extended to address another important problem,
the estimation of a likelihood function in situations where that function cannot be well-approximated
by an analytical form. One is often faced with this situation when performing statistical inference
with data from the sciences, due the complexity of the data and of the processes that generated those data.
We emphasize applications where
using existing likelihood-free methods of inference
would be challenging due to the high dimensionality of the sample space, but where our 
spectral series method yields a reasonable estimate 
of the likelihood function.
We provide theoretical guarantees and illustrate the effectiveness of our proposed method with numerical experiments.
\end{abstract}

\section{INTRODUCTION}
\label{sec::intro}

There has been growing interest in the problem of estimating the ratio of two probability densities, 
$\beta(\x) \equiv f(\x) / g(\x)$, given $i.i.d.$~samples from unknown distributions $F$ and $G$.
For example, these ratios play a key role in matching training and test 
data in so-called transfer learning or domain adaptation \citep{SugiyamaSuzuki}, where the goal is to predict an outcome $y$ given test data $\x$ 
from a distribution ($G$) that is different from that of the training data ($F$).
Estimated density ratios also appear in novelty detection \citep{hido2011statistical},
conditional density estimation \citep{sugiyama2010conditional}, selection bias correction \citep{Gretton}, and classification \citep{nam2012computationally}.

Experiments have shown that it is suboptimal to estimate $\beta(\x)$ by first estimating the two
component densities and then taking their ratio \citep{SugiyamaImportance}.
Hence, several alternative approaches have been proposed that directly estimate $\beta(\x)$; e.g.,
 \emph{uLSIF}, an estimator obtained via
least-squares minimization \citep{Kanamori};  \emph{KLIEP}, which is obtained 
via Kullback-Leibler divergence minimization \citep{SugiyamaImportance}; \emph{KuLSIF}, a kernelized version of \emph{uLSIF} \citep{kanamori2012statistical};
and \emph{kernel mean matching}, which is based on minimizing the mean discrepancy between
transformations of the two samples in a Reproducing Kernel Hilbert Space (RKHS) \citep{Gretton}. 
For a review of techniques see \citet{Margolis}.

Existing methods are not effective when $\x$ is of high dimension, and hence authors
recommend a dimension reduction prior to implementation
\citep{sugiyama2011directDensityRatio}. 
As is the case with 
any data reduction, such a step can result in significant loss of information.
Here we propose a novel series estimator of $\beta(\x)$ 
designed to take advantage of the intrinsic dimensionality of $\x$, but without 
an explicit dimension reduction step. 
Thus, this work addresses a critical need by constructing
a nonparametric estimator for $\beta(\x)$ which performs well even when $\x$ is of high dimension.

Our method is fast and simple to implement. 
To approximate $\beta(\x)$, we expand $\beta$ in terms of the eigenfunctions of a kernel-based operator. These eigenfunctions are orthogonal with respect to the underlying data distribution as opposed to the Lebesgue measure on the ambient space. In fact, the eigenfunctions form a Fourier-like basis adapted to the submanifold structure with the low-order components smoother than the higher-order ones, see Figure \ref{fig::eigenSpiral}.
 As we shall see, this basis is particularly well-suited for approximating smooth functions in high dimensions.
Unlike 
RKHS methods \citep{Gretton} (which do not explicitly compute the eigenfunctions themselves), our approach to nonparametric density estimation allows for out-of-sample extensions and a principled way of choosing tuning parameters via well-studied model selection techniques such as cross-validation.

\comment{Because these eigenfunctions are orthogonal with respect to the data distribution rather than the Lebesgue measure, they 
take into account the underlying
geometry of the data, leading to better estimates in high dimensions. 
In fact, these eigenfunctions
behave like a Fourier basis \emph{adapted to the submanifold structure}
with the low-order components smoother than the higher-order ones, see Figure \ref{fig::eigenSpiral}.
 This basis is well-suited for approximating smooth functions in high dimensions.
Finally, unlike 
RKHS methods \citep{Gretton}, our technique allows for out-of-sample extensions, and a principled way of choosing tuning parameters via cross-validation and other well-studied model selection techniques.
}

\begin{figure}
  \centering
  \vspace{.2in}
\includegraphics[page=5,scale=0.15]{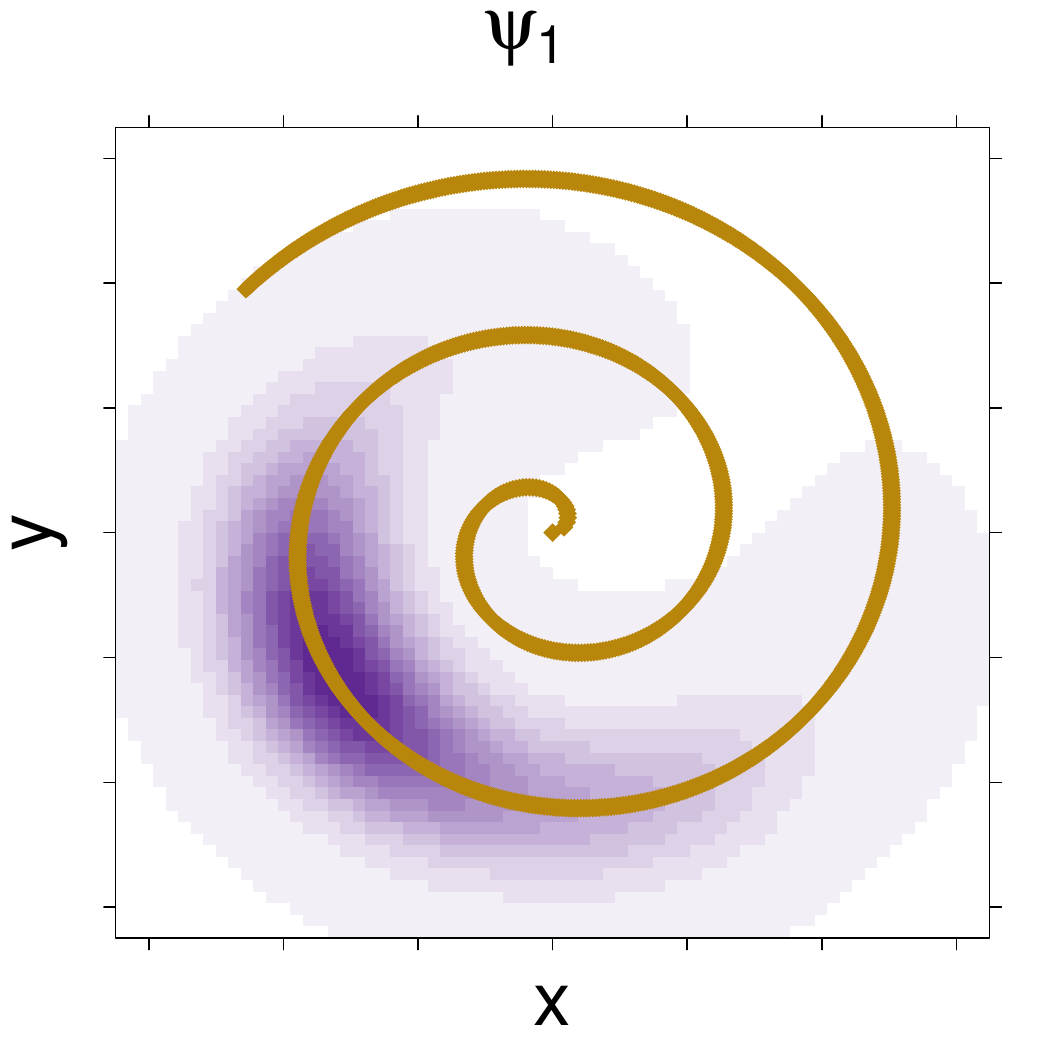} \hspace{-2.mm}
\includegraphics[page=6,scale=0.15]{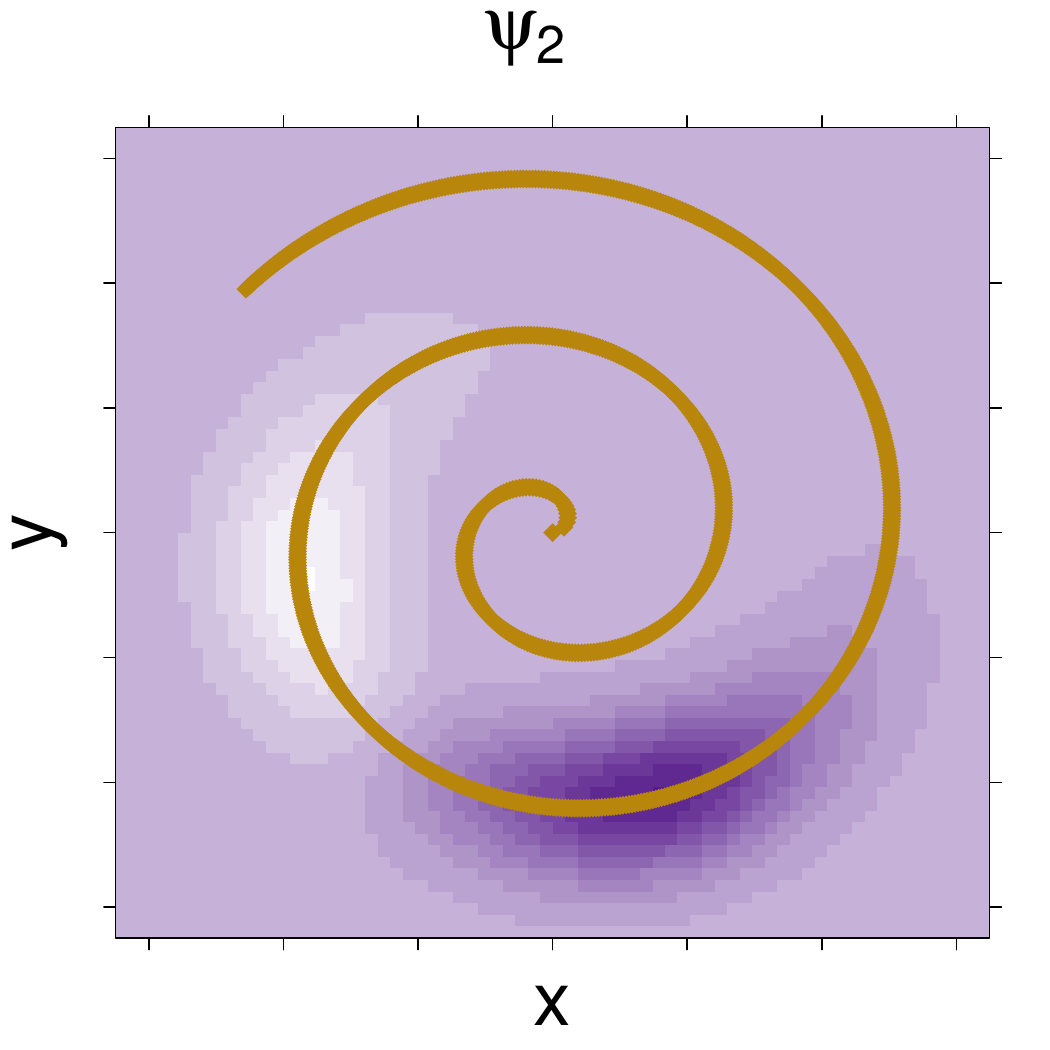} \hspace{-2.mm}
\includegraphics[page=7,scale=0.15]{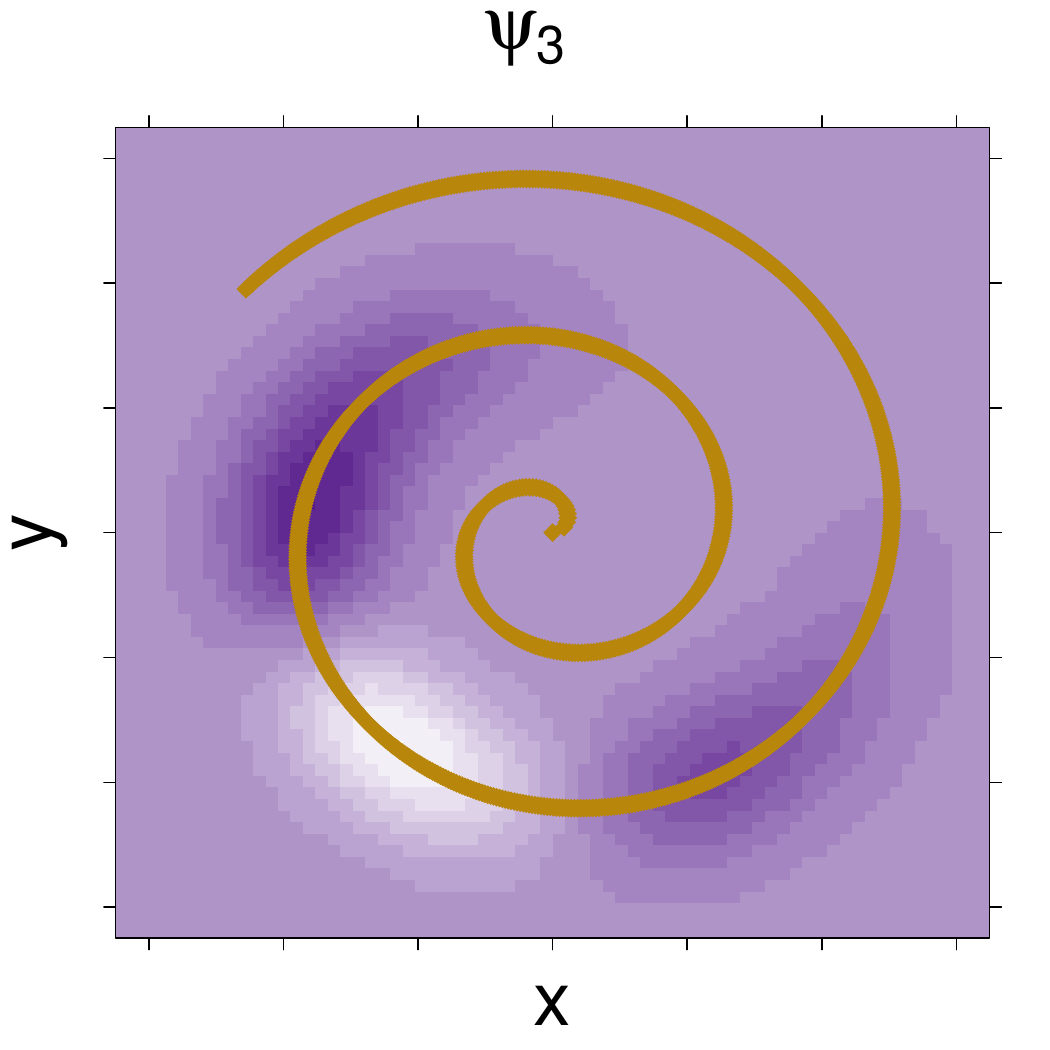} \hspace{-2.mm}
\vspace{.2in}
  \caption{{\small Level sets of the top three eigenfunctions of a kernel-based operator when the support of the data $(x,y)$ is close to a spiral. The eigenfunctions
  form a Fourier-like basis adapted to the geometry of the data. This basis is well-suited for approximating smooth functions in the region around the spiral.}}
  \label{fig::eigenSpiral}
\end{figure}

We extend our proposed methodology for estimating density ratios to the problem
of estimating the likelihood function of observing data $\x$ given parameters $\theta$. 
Estimation of the likelihood is necessary when the complexity of the data-generation process
prevents derivation of a sufficiently accurate analytical form for the likelihood function.
Here we exploit the fact that, in many such situations, one can
{\em simulate} data sets $\x$ under different parameters $\theta$.
This is often the case in statistical inference problems in the sciences, 
 where the relationship between parameters of interest and observable data is complex, 
but accurate simulation models are available; see, for example, genetics (\citealt{beaumont2010approximate,estoup2012estimation}) and  astronomy (\citealt{cameron2012approximate,weyant2013likelihood}).
Problems of this type have motivated recent interest in methods of {\it likelihood-free inference}, 
which includes methods of {\it Approximate Bayesian Computation (ABC)}; see \citet{marin2012approximate} 
for a review. 
\comment{The goal of ABC is to draw a 
sample from an \emph{approximation} of the Bayesian posterior distribution $f(\theta|\x)$. Unfortunately, it
requires non-trivial tuning parameters to be chosen, and,
if $\x$'s are high-dimensional,
it is unavoidable to use problem-specific lower-dimensional summary statistics of the data to implement it.
These statistics are not always easy to be designed, and one 
usually incurs in loss of information \citep{blum2013comparative}. 
In order to avoid these issues of ABC, we propose estimating the likelihood function directly.
\citet{diggle1984monte} and \citet{fan2013approximate} also proposed likelihood estimation as an alternative
to traditional ABC, however the emphasizes was not on high-dimensional data, where traditional estimation methods fail.}

In our implementation, we redefine the likelihood function
as $\L(\x;\theta) \equiv f(\x|\theta)/g(\x)$, where 
$g(\x)$ is a density with support larger than that of $f(\x |\theta)$.
This formulation differs from the standard 
definition of the likelihood by only a multiplicative term which is constant in $\theta$,
and hence $\L(\x;\theta)$ can still be used for likelihood-based 
inference (including maximum likelihood estimation). In particular, the shape of the posterior for $\theta$
is unaffected. 
The challenge of estimating the likelihood is now a density ratio estimation
problem.
This approach will yield significant advantages in cases where $g$ is chosen to
focus high probability on the low-dimensional subspace in which the data $\x$ lie. One natural choice is $g(\x) = \int f(\x | \theta) d\pi(\theta)$,
where $\pi$ is a well-chosen prior distribution for $\theta$.
The orthogonality of the spectral series with respect to $g$ results in an efficient
implementation of the estimator. 
Moreover, 
directly estimating the ratio $f(\x|\theta)/g(\x)$ may itself be easier than 
estimating $f(\x|\theta)$,
e.g., when the conditional distributions $f(\x|\theta)$ for different $\theta$ are similar\footnote{A trivial example: If $\x$ is independent of $\theta$,
 $f(\x|\theta)/g(\x)=1$ is a constant function, whereas $f(\x|\theta)=f(\x)$ may be a harder to estimate (nonsmooth) function.}, or 
 when they have similar support in high dimensions.  To our knowledge, this is the first work that proposes a spectral series approach to non-parametric density estimation and likelihood inference in high dimensions.

 The organization of the paper is as follows:
In Section \ref{sec::ratio} we present our density ratio estimator and apply it to a prediction problem in astronomy with covariate shift.
In Section \ref{sec-methCond}, we show how our method can be extended to 
estimating a likelihood function, and provide experiments that show its advantages over traditional methods.  
Finally, in Section \ref{sec-theory}, we provide theoretical guarantees and rates of convergence of the proposed
estimators. Full proofs as well as details on the astronomy data are provided in Supplementary Materials. 

\section{SPECTRAL SERIES ESTIMATOR OF A DENSITY RATIO}
\label{sec::ratio}

In this section we will present the mathematical details behind our
spectral series estimator of a density ratio.
To begin, let $\vec{x}$ denote a $d$-dimensional random vector, assumed to lie in the subspace
$\mathcal{X}$.
We observe 
an $i.i.d.$~sample
$\vec{x}^F_1,\ldots,\vec{x}^F_{n_F}$ from an unknown distribution $F$, as well as 
an $i.i.d.$~sample
$\vec{x}^G_1,\ldots,\vec{x}^G_{\nsubG}$ from an unknown distribution $G$. 
The goal is to estimate 
$$\beta(\vec{x}) \equiv f(\x)/g(\x).$$

We assume that $F\hspace{-1.1mm} \ll \hspace{-1.1mm}G$ so that this ratio is well-defined.

Let $K_\x(\vec{z},\vec{y})$ be a bounded, symmetric, and positive definite kernel\footnote{In 
our applications, we use the Gaussian kernel $K_\x(\vec{z},\vec{y}) =
\exp(-d^2(\vec{z},\vec{y})/4\epsilon)$, where $d(\cdot,\cdot)$ 
is the Euclidean distance in $\Re^d$.}, and let $\{\psi_j\}_{j\in \mathbb{N}}$ be the eigenfunctions of the operator 
$\textbf{K}_\x: L^2(\mathcal{X},G) \longrightarrow L^2(\mathcal{X},G)$ \citep{Rosasco}:
\begin{align}
\label{eq::kpcaOperator}
\textbf{K}_\x(h)(\vec{z})=\int_\mathcal{X} K_\x(\vec{z},\vec{y})h(\vec{y})dG(\vec{y}). 
\end{align}

Our spectral series estimator relies on the fact that $\{\psi_j\}_{j\in \mathbb{N}}$ is an orthonormal basis
of $L^2(\mathcal{X},G)$, i.e., the eigenfunctions are orthonormal with respect to the data distribution $G$ rather than the Lebesgue measure:
\begin{align*}
\int_\mathcal{X} \psi_i(\vec{x})\psi_j(\vec{x})dG(\vec{x})=\I(i=j). 
\end{align*}

Hence, for $\beta(\vec{x}) \in L^2(\mathcal{X},G)$, we can write
\begin{align}
\label{eq::expBeta}
\beta(\vec{x})=\sum_{j \in \mathbb{N}} \beta_j \psi_j(\vec{x}), 
\end{align}
where $\beta_j=\int \beta(\vec{x}) \psi_j(\vec{x})dG(\vec{x})=\E_{F}[\psi_j(\vec{X})].$

Since $G$ is unknown, the $\psi_j$'s must be estimated. First, 
we compute the $J\leq n_G$ eigenvectors $\widetilde{\psi}_1,\ldots,\widetilde{\psi}_J$  (with largest eigenvalues)
of the Gram matrix 
based on the sample from $G$,
 \begin{align*}
\left[ K_\x\left(\vec{x}^G_i,\vec{x}^G_j\right) \right]_{i,j=1}^{n_G}.  
 \end{align*}
 

These functions are then extended to all $\x \in \mathcal{X}$ via the Nystr\"om Extension \citep{Drineas05onthe}
\begin{align*}
 \widehat{\psi}_j(\vec{x})=\frac{\sqrt{n_G}}{\widehat{\ell}^{\x}_j}\sum_{k=1}^{n_G} 
    \widetilde{\psi}_j\left(\vec{x}^G_k\right) K_\x\left(\vec{x},\vec{x}^G_k\right),
\end{align*}
where $\widehat{\ell}^{\x}_j$ is the eigenvalue associated to the 
eigenvector $\widetilde{\psi}_j$.
Next we estimate the $\beta_j$'s in Eq.~(\ref{eq::expBeta}) using the sample from $F$:
\begin{align*}
\widehat{\beta}_j = \frac{1}{n_F}\sum_{k=1}^{n_F}\widehat{\psi}_j\left(\vec{x}^F_k\right).
\end{align*}
Our spectral series estimator is finally given by
\begin{align}
\label{eq-final}
\widehat{\beta}(\x) = \left(\sum_{j=1}^{J} \widehat{\beta}_j \widehat{\psi}_j(\x)\right)_{\!+}.
\end{align}

This approach can be motivated as follows.
The basis functions $\widehat{\psi}_j$ are consistent estimators of the 
eigenfunctions $\psi_j$ of the corresponding integral operator \citep{Bengio}. 
In our  nonparametric model, $J$ is a tuning parameter that controls 
the bias/variance tradeoff: Decreasing $J$ decreases the variance, 
but increases the bias of the estimator. We choose $J$ (and the other tuning parameters) 
in a principled way described below. 

\vspace{2mm}
\textbf{Model Selection.} To evaluate the performance of an estimator $\widehat{\beta}(\x)$, we use the loss function
\begin{align*}
L(\widehat{\beta},\beta) &\equiv \int \left(\widehat{\beta}(\vec{x})-\beta(\vec{x}) \right)^2dG(\vec{x}) \\
&=\int \widehat{\beta}(\vec{x})^2 dG(\vec{x})-2\!\int \widehat{\beta}(\vec{x})dF(\vec{x})+K
\end{align*}

where $K$ does not depend on $\widehat{\beta}$. We estimate this quantity (up to $K$) using
\begin{align}
\label{eq-lossBetaEst}
\widehat{L}(\widehat{\beta},\beta) =\frac{1}{\widetilde{n}_G}\sum_{k=1}^{\widetilde{n}_G}\widehat{\beta}^2(\widetilde{\vec{x}}_k^G)-\frac{2}{\widetilde{n}_F}\sum_{k=1}^{\widetilde{n}_F}\widehat{\beta}(\widetilde{\vec{x}}_k^F),
\end{align}
where $\widetilde{\x}^G_1,\ldots,\widetilde{\x}^G_{\widetilde{n}_G}$ is a validation sample from $G$, and 
$\widetilde{\x}^F_1,\ldots,\widetilde{\x}^F_{\widetilde{n}_F}$ is a validation sample from $F$.
Tuning parameters are chosen to minimize $\widehat{L}(\widehat \beta, \beta)$. 
Note that because of the orthogonality
of the $\widehat{\psi}_j$, it is not necessary 
to recompute the estimated coefficients $\widehat{\beta}_j$'s for each value of $J$ in Eq. (\ref{eq-final}), unlike most estimation
procedures, where estimated coefficients have to be recomputed for each configuration of the tuning parameters. 
In other words, 
only the tuning parameters associated with
the kernel (in our case, the kernel bandwidth $\epsilon$) affect the computation time.
%
%
%

\subsection{Application: Correction to Covariate Shift in Photometric Redshift Prediction}
\label{sec-CS}
Assume we observe a sample 
of unlabeled data, as well as a sample
of labeled data, where the $Z$'s represent the labels and $\vec{x}$'s are the covariates. One is often interested
in estimating the regression function $\E[Z|\vec{x}]$  under selection bias, i.e., in situations where the distributions of labeled and unlabeled samples 
($f_L(\x)$ and $f_U(\x)$, respectively) are different. 
If the estimate $\widehat{\E}[Z|\vec{x}]$
is constructed using the labeled data with the goal of 
 predicting $Z$ from $\vec{x}$ on the \emph{unlabeled} data, corrections have to be made. 
A key quantity for making this correction under the \emph{covariate shift assumption} \citep{Shimodaira2000} is 
the density ratio $f_U(\x)/f_L(\x)$, the so-called \emph{importance weights} \citep{Gretton}. 
We now compare various estimators of 
importance weights for a key problem in astronomy, namely that of
redshift estimation \citep{Sheldon}.


\begin{figure}[ht]
\vspace{.0in}
\centerline{\includegraphics[width=0.38\textwidth]{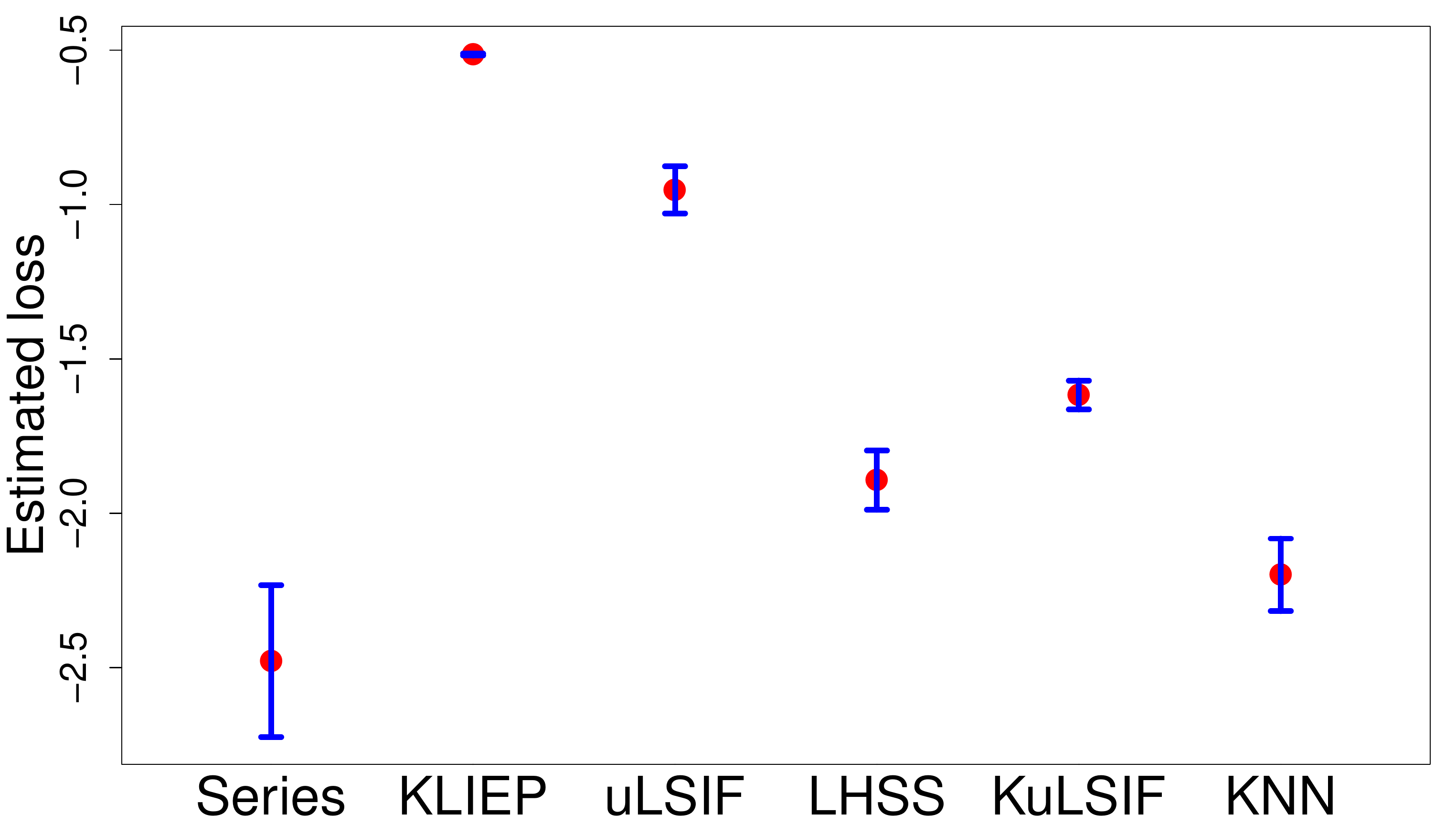}}
\vspace{.15in}
\caption{{\small Estimated losses of $\widehat{\beta}(\x)$  with standard errors for SDSS data. The spectral series
  estimator has best performance.}}
  \label{fig::photoZ}
\end{figure}

We use a subset of the Sloan Digital Sky Survey  \citep{aihara2011eighth}. 
The ultimate goal is to build
a predictor of galaxy redshift $Z$ based on photometric data $\x$; see Supplementary Materials
for details. We are given a training set with covariates $\x$ of galaxies and their redshifts, as well as unlabeled target data. Because it is 
difficult to acquire the true redshift of faint galaxies, these data suffer from selection bias. 
We compare our method 
of estimating the importance weights (\emph{Series}) to 
\emph{uLSIF}, \emph{KLIEP}, and \emph{KuLSIF}, described in Section \ref{sec::intro}. We also
compute \emph{LHSS}, which uses \emph{uLSIF} after applying a dimension reduction technique specifically designed
for estimating a density ratio, see \citet{sugiyama2011directDensityRatio}.
Moreover, we include a
comparison with a $k$-nearest neighbors estimator (\emph{KNN}) proposed in the astronomy literature (\citealt{Lima1}),
which it not based on ratios.
We do not show results of ratio-based estimators  because 
the estimates of $f_L(\x)$ are close to zero for many $\x$'s, inducing estimates of $\beta$ that are infinity.

Figure \ref{fig::photoZ}
shows the estimated losses of the different methods of estimating the ratio $f_U(\x)/f_L(\x)$ when using 5,000 labeled and 5,000  unlabeled samples,
and 10 photometric covariates $\x$. We use 60\% of the data for training, 20\%  for validation and 20\% for testing. Even though this example has 
a covariate space with as few as 10 dimensions, we can already see the benefits 
of the spectral series estimator. 

\section{EXTENSION: SPECTRAL SERIES ESTIMATOR OF A LIKELIHOOD FUNCTION}
\label{sec-methCond}

Our framework for estimating density ratios can be extended to the problem of 
estimating a high-dimensional likelihood function.
To the derivation described above, we add
$\theta \in \Theta$, the $p$-dimensional parameter.
In this context, $\x \in \mathcal{X} \subseteq  \Re^d$ is a random vector representing a single sample observation.
We will adopt a Bayesian perspective, and let 
$F_\theta$ be the marginal distribution for $\theta$, i.e., the prior, and let $G$ denote the marginal distribution
for $\x$. Then,
let $(\x_1^F, \theta_1),\ldots,(\x_{n_F}^F, \theta_{n_F})$ be an $i.i.d.$~sample from the joint distribution of
$\x$ and $\theta$.
Further, let $\vec{x}^G_1,\ldots,\vec{x}^G_{n_G}$ be an $i.i.d.$~sample from $G$.
Our objective is to estimate the ratio
\begin{align}
\label{eq::defLikeli}
 \L(\x;\theta) \equiv \frac{f(\x|\theta)}{g(\x)},
\end{align}
where $f(\x | \theta)$ is the conditional density of $\x$ given $\theta$, and $g(\x)$ is the marginal density
for $\x$. This is, up to a multiplicative factor that is not a function of $\theta$, the standard
definition of the likelihood function.

To estimate $\L(x;\theta)$, we use a spectral series approach as before, but because the likelihood is a function of both $\x$ and $\theta$, we consider 
the \emph{tensor product} of a basis for $\x$ and a basis for $\theta$,
$\{\Psi_{i,j}\}_{i,j\in \mathbb{N}}$, where
\begin{align*}
\Psi_{i,j}(\x,\theta)= \psi_j(\vec{x})\phi_i(\theta),\ i,j\in\mathbb{N}. 
\end{align*}
The construction of the separate bases $\{\psi_j\}_j$ and $\{\phi_i\}_i$ proceeds just
as described in Section \ref{sec::ratio}. 
Note that for $\theta$, we
consider the eigenfunctions $\{\phi_i\}_i$ 
of the operator $\textbf{K}_\theta\!: L^2(\Theta,F_\theta) \longrightarrow L^2(\Theta,F_\theta)$: 
\begin{align*}
\textbf{K}_\theta(h)(\xi)=\int_\Theta K_\theta(\xi,\mu)h(\mu)dF_\theta(\mu),
\end{align*}
where $K_\theta$ is not necessarily the same kernel as $K_{\x}$.
That is, while $\{\psi_j\}_j$ is estimated using a Gram matrix based on $\vec{x}^G_1,\ldots,\vec{x}^G_{n_G}$,
$\{\phi_i\}_i$ is estimated using $\theta_1,\ldots, \theta_{n_F}$.

Since $\{\phi_i\}_i$ is an orthonormal basis of functions in $L^2(\Theta,F_\theta)$, 
 the tensor product $\{\Psi_{i,j}\}_{i,j}$ is an orthonormal basis 
for functions in $L^2(\Theta\times\mathcal{X},F_\theta \times G)$. 

The projection of $\L(\x;\theta)$ onto $\{\Psi_{i,j}\}_{i,j}$
is given by 
\begin{align*}
   \sum_{i,j\in \mathbb{N}}\beta_{i,j}\Psi_{i,j}(\x,\theta),
\end{align*}
where 
\begin{align}
\label{eq::betaExpre}
\beta_{i,j}&=\iint \L(\x;\theta) \Psi_{i,j}(\x,\theta) dG(\x)dF_\theta(\theta)  \notag \\ 
&=\E_F\!\left[\Psi_{i,j}(\x,\theta)\right].  
\end{align}

Hence, we define our likelihood function estimator by
\begin{align*}
\widehat{\L}(\x;\theta) = \sum_{i=1}^{I} \sum_{j=1}^{J} \widehat{\beta}_{i,j} \widehat{\Psi}_{i,j}(\x,\theta),
\end{align*}
where
$\widehat{\beta}_{i,j} = \frac{1}{n_F}\sum_{k=1}^{n_F}\widehat{\Psi}_{i,j}\left(\vec{x}^F_k,\theta_k\right)$, and

$$\widehat{\Psi}_{i,j}(\x,\theta)=\widehat{\psi}_j(\vec{x})\widehat{\phi}_i(\theta)$$
is the estimator for $\Psi_{i,j}(\x,\theta)$, obtained via a Nystr\"om extension as in Section \ref{sec::ratio}.
The tuning parameters $I$ and $J$ control the bias/variance 
tradeoff. Because we define the likelihood in terms of $g(\x)$ 
(Eq. \ref{eq::defLikeli}), we can take advantage of the 
orthogonality of the basis functions 
when estimating the coefficients $\beta_{i,j}$; see Eq. \ref{eq::betaExpre}. 
The result is a simple and fast-to-implement procedure for estimating
likelihood functions for high-dimensional data.

\vspace{2mm}
\textbf{Model Selection.} To evaluate the performance of a given estimator, we use the loss function
\begin{eqnarray}
\label{eq::lossL2L}
L\left(\widehat{\L},\L\right) &\equiv& \int \left(\widehat{\L}(\vec{x};\theta)-\L(\vec{x};\theta) \right)^2dG(\vec{x})dF(\theta) =\notag  \\
&=& \int \widehat{\L}(\vec{x};\theta)^2 dG(\vec{x})dF(\theta) - \notag \\
&-&2\!\int \widehat{\L}(\vec{x};\theta)dF(\theta,\vec{x})+K,
\end{eqnarray}
where $K$ does not depend on $\widehat{\L}$.  We can
estimate this quantity (up to $K$) by
\begin{align*}
\widehat{L}\left(\widehat{\L},\L\right) =&\frac{1}{B}\sum_{l=1}^B\left[  \frac{1}{\widetilde{n}}\sum_{k=1}^{\widetilde{n}}\left( \widehat{\L}\left(\widetilde{\x}^G_k|\widetilde{\theta}^{(l)}_k\right)\right)^2\right]  \\
& -\frac{2}{\widetilde{n}}\sum_{k=1}^{\widetilde{n}}\widehat{\L}(\widetilde{\vec{x}}_k^F|\widetilde{\theta}_k),
\end{align*}
where $\widetilde{\x}^G_1,\ldots,\widetilde{\x}^G_{\widetilde{n}}$ is a validation sample from $G$; 
$(\widetilde{\theta}_1,\widetilde{\x}^F_1),\ldots,(\widetilde{\theta}_{\widetilde{n}},\widetilde{\x}^F_{\widetilde{n}})$ 
is a validation sample from the joint distribution of $\x$ and $\theta$;
 $\widetilde{\theta}_1^{(l)},\ldots,\widetilde{\theta}_{\widetilde{n}}^{(l)}$ for $l=1,\ldots,B$
are random permutations of the original sample $\widetilde{\theta}_1,\ldots,\widetilde{\theta}_{\widetilde{n}}$; and $B$ 
is a number limited only by computational considerations.
We choose tuning parameters so as to minimize $\widehat{L}$. 

\textbf{Remarks.} 
 By choosing an appropriate kernel, the spectral series approach can be extended to discrete data. For example, in 
 \citet{lee2010spectral}, p. 185, the authors suggest a distance kernel that take into account the discrete nature of genetic SNP data.
 By minimizing the loss in Eq. (\ref{eq::lossL2L}), one can also select the best kernel from a set of reasonable candidate kernels.
 Finally, our procedure can be scaled up to large sets of simulated data by speeding up 
the eigendecomposition of the Gram matrix. Possible methods include the Nystr\"om extension \citep{Drineas05onthe} and procedures described in, e.g., \citet{belabbas2009spectral} and \citet{Halko}. In particular, some of these approaches can be parallelized for even higher computational efficiency.

\subsection{Numerical Experiments}
\label{sec::abc}
Estimation of a likelihood function 
is of particular
value in cases where the complexity of the data and the data-generating process prevents construction of a sufficiently accurate analytical form,
a situation typically present in high-dimensional scientific
data. The general setup is as follows. We have data which are modeled as 
an $i.i.d$.~sample $\x_1,\ldots,\x_m$ from the distribution 
$f(\x|\theta)$. Our goal is to infer the value $\theta$.
Although we are able to simulate from $f(\x|\theta)$ for fixed $\theta$,
we lack an analytical form for the likelihood function. Hence, we use the methodology
of Section \ref{sec-methCond} to estimate $\L(\x;\theta)$ from a \emph{simulated sample}. 
Once we have an estimate
$\widehat{\L}(\x;\theta)$, we can approximate the likelihood of an \emph{observed sample}
according to
\[
 \widehat{\L}\left((\x_1,\ldots,\x_m);\theta\right)=\prod_{k=1}^m \widehat{\L}(\x_i;\theta).  
\]
This approximation can then be used in likelihood-based inference by, for example, plugging the expression into Bayes Theorem or by finding the maximum likelihood estimate.

In what follows we present five numerical examples where the ambient dimensionality of $\vec{x}$
is larger than its intrinsic dimensionality. In all experiments, we choose a uniform 
prior distribution on the parameter space.

\vspace{2mm}
\textbf{Spiral.} The data are $i.i.d.$ observations of $(X^{(1)},X^{(2)})$, where
 $$X^{(1)}=\theta\cos{\theta}+N(0,1)\mbox{ and } X^{(2)}=\theta\sin{\theta}+N(0,1)$$
for  $0<\theta<15$.
Although the dimension of the sample space is 2, the data lie close to a one-dimensional spiral.

\vspace{2mm}
\textbf{Klein Bottle.}  The data are $i.i.d.$ observations of $(X^{(1)},X^{(2)},X^{(3)},X^{(4)})$, where
\[
  \left\{ 
  \begin{array}{l}
    X^{(1)}=2(\cos{\theta_2}+1)\cos{\theta_1}+N(0,1)\\
    X^{(2)}=2(\cos{\theta_2}+1)\sin{\theta_1}+N(0,1)\\
    X^{(3)}=2\sin{\theta_2}\cos{\theta_1/2}+N(0,1)\\
    X^{(4)}=2\sin{\theta_2}\sin{\theta_1/2}+N(0,1)
  \end{array} \right.
\]
for $0<\theta_1,\theta_2<2\pi$.
The dimension of the sample is 4, but the data lie close to a two-dimensional Klein Bottle
embedded in $\Re^4$.

\vspace{2mm}
\textbf{Transformed Images.} In this example, we rotate and translate an image of a tiger, see the top row of Figure \ref{fig::tigerGalaxies}.
The model parameters are $(\theta, \rho_x, \rho_y)$. The transformed images are centered at $(\rho_x+N_T(0,10),\rho_y+N_T(0,10))$\footnote{$N_T$ is the truncated normal to guarantee that the parameters are in the range of the image.} with rotation angle $(\theta+N(0,10))$.
The final images are cropped to $20 \times 20$ pixels,
i.e., the sample space has dimension $400$.
\begin{figure}[lineheight]
\vspace{.0in}
\begin{center}
Transformed Images
\centerline{\includegraphics[scale=0.9]{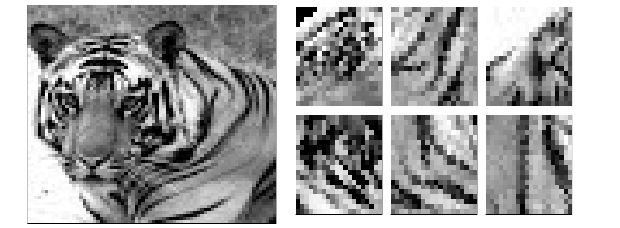}}\vspace{1mm}
Edges
\centerline{\includegraphics[page=1,scale=0.9]{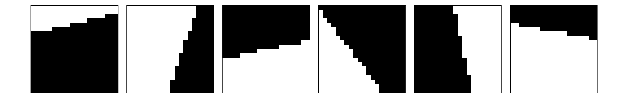}}\vspace{1mm}
Galaxies\vspace{-1.9mm}
\centerline{\includegraphics[page=4,scale=0.129]{galaxiesGenerate3TransformationNoise.pdf}\hspace{-4mm}
\includegraphics[page=8,scale=0.129]{galaxiesGenerate3TransformationNoise.pdf}\hspace{-4mm}
\includegraphics[page=12,scale=0.129]{galaxiesGenerate3TransformationNoise.pdf}} \vspace{1mm} 
\end{center}
\vspace{-.2in}
\caption{ {\small Some examples of data generated according to Section~\ref{sec::abc}. (The top left image is the original image in ``Transformed Images''.)}}
 \label{fig::tigerGalaxies}
\end{figure}

\vspace{2mm}
\textbf{Edges.} Here we generate $20 \times 20$ images of binary edges from a model with two parameters, $\alpha$ and $\lambda$.
The data are $i.i.d.$ observations of an edge with rotation angle $\alpha+N(0,\pi/4)$ and displacement $\lambda+N_T(0,0.5)$ from the center,  see Figure \ref{fig::tigerGalaxies}
for some examples.

\vspace{2mm}
\textbf{Simulated Galaxy Images.} The last example is a simplified version of a key estimation problem in astronomy, namely that of shear estimation \citep{bridle2009handbook}.  We use the GalSim Toolkit\footnote{https://github.com/GalSim-developers/GalSim}
to simulate realistic galaxy images. We sample two parameters: First, the orientation with respect to the $x$-axis of the image and, second, the axis ratio of the galaxies, which measures their ellipticity. To mimic a realistic situation, the observed data are low-resolution images of size $20 \times 20$.  Figure \ref{fig::tigerGalaxies}, bottom, shows some examples. These images have been degraded by observational effects such as background noise, pixelization, and 
blurring due to the atmosphere and telescope;
see the Supplementary Materials for details.

We assume that the orientation and axis ratio of galaxy $i$ are given by 
$a_i \sim \mbox{Laplace}(\alpha,10) \mbox{ and } r_i \sim N_T(\rho,0.1^2),$
respectively. We seek to infer $\theta=(\alpha,\rho)$ based on an observed 
$i.i.d.$~sample of images $\x$ contaminated by observational effects.
Notice we do not observe
$a_i$ and $r_i$, but only $\vec{x}_i$, the 400-dimensional noisy image.

\vspace{2mm}
\textbf{Methods.} 
In the examples above, the likelihood function is estimated based on $n_F=n_G=5,000$  observations from the simulation model; 
60\% of the data are used for training and 40\% for validation. 
We compare \emph{Series}, our spectral series estimator 
from Section \ref{sec-methCond}, with two state-of-the-art estimators of $f(\vec{x}|\theta)$.
The first estimator is \emph{KDE} -- a kernel density estimator based on taking the ratio of kernel estimates
of $f(\vec{x},\theta)$ and $f(\theta)$. 
We use the implementation
from the package ``np'' \citep{np} for R. For the Spiral and Klein bottle examples, 
we select the bandwidths of KDE via cross-validation. However, the high dimension of the other examples (Transformed, Edges, and Galaxy) makes a cross-validation approach computationally intractable and the density estimates numerically unstable.
For these examples, we instead use the default reference rule for the bandwidth, and we reduce the dimensionality of the data with PCA with number
of components chosen by minimizing the estimated loss (Eq.~\ref{eq::lossL2L}).
The second estimator in our comparisons is  \emph{LS} --
the direct least-squares conditional density estimator of \citet{sugiyama2010conditional}.
This estimator is based on 
a direct expansion of the likelihood into a set of prespecified 
functions. This approach typically yields better results than estimators based on the ratio of random variables. 
Again, to avoid the problem of high dimensionality in the examples with $d>4$, 
we implement \emph{PCA+LS},
the direct least-squares conditional density estimator after dimension reduction via PCA, with number
of components chosen so as to minimize the estimated loss. PCA has the additional goal of decorrelating adjacent pixels 
in the images examples.
\begin{table*}[ht]
\caption{\small Estimated $L^2$ loss (with standard errors) of the likelihood function estimators. Best-performing models with smallest loss are in bold fonts. } \label{table::loss}
\vspace{1mm}
\centering
\tabcolsep=0.07cm
{\small
\begin{tabular}{l|c|ccccc}
{\bf DATA}  & {\bf DIM.}  &\multicolumn{5}{c}{{\bf $\mathbf{L^2}$ LOSS}}       \\
\hline 
 & & {\it Series }& {\it LS} & {\it PCA+LS} & {\it KDE} & {\it PCA+KDE}\\
Spiral         &2 &7.13 (0.14)& 6.61 (0.12)& ---  & \textbf{2.95 (0.30)} &   ---   \\
Klein Bottle             & 4&\textbf{1.45 (0.07)} & 2.02 (0.06)& --- &1.68 (0.11) &  ---   \\
Transf. Images             & 400&\textbf{20.94 (0.03)} &26.91 (0.04)&27.12 (0.03) &---   & 26.62 (0.06)\\
Edges              &  400& \textbf{0.70 (0.03)} & 1.77 (0.02) &1.55 (0.03)& ---  &  1.60(0.02)  \\
Galaxy Images             &400 & \textbf{40.94 (0.03)}& 42.57 (0.01)&42.53 (0.01) & ---     &  43.99 (0.04)
\end{tabular}}
\end{table*}

\begin{table*}[ht]
\caption{\small  Estimated average likelihood (with standard errors) of the likelihood function estimators. Best-performing models with largest average likelihood are in bold fonts. } \label{table::likeli}
\vspace{1mm}
\centering
\tabcolsep=0.07cm
{\small
\begin{tabular}{l|c|ccccc}
{\bf DATA}  & {\bf DIM.}       & \multicolumn{5}{|c}{{\bf AVERAGE LIKELIHOOD}}     \\
\hline 
 & & {\it Series }& {\it LS} & {\it PCA+LS} & {\it KDE} & {\it PCA+KDE}\\
Spiral         &2 &    16.54 (0.16)& 19.49 (0.14) & --- & \textbf{28.62 (0.01)}  & --- \\
Klein Bottle             & 4 &5.62 (0.08) & 4.96 (0.08) & ---&  \textbf{5.63 (0.13)}   & --- \\
Transf. Images             & 400&  \textbf{8.31 (0.03)} &1.83 (0.03)&1.08 (0.02) &  ---  & 1.58 (0.06)\\
Edges              &  400&   \textbf{3.69 (0.04)} & 1.72 (0.02)&2.55 (0.03) & ---   & 2.10 (0.02)\\
Galaxy Images             &400 &   \textbf{4.63 (0.04)} & 2.24 (0.01)&2.43 (0.02) & ---  &1.01 (0.04)
\end{tabular}}
\end{table*}

\vspace{2mm}
\textbf{Results.} 
In Tables \ref{table::loss} and \ref{table::likeli}, we present the estimated $L^2$ loss (Eq. \ref{eq::lossL2L}),  as well as
the estimated average likelihood $\E_{(\vec{X},\theta)}[\widehat{\L}(\vec{X};\theta) ]$ based on a test set with 3,000 observations\footnote{To make results
comparable, we renormalize the estimated likelihood functions to integrate to 1 in $\theta$.}. 
Both measures indicate 
that, while traditional methods have better performance in low dimensions, our spectral series method yields substantial improvements
when the ambient dimensionality of the sample space is large. 
Note that even after dimension reduction, \emph{LS}
does not yield the same performance as \emph{Series}. In fact, in some cases, a dimension reduction via PCA leads to less accurate estimates.
%

As a further illustration, Figure \ref{fig::posteriors} shows the estimated likelihood function 
for a sample of size $m=10$  drawn from the galaxy image model with parameters $\alpha=80^{\circ}$ and $\rho=0.2$.
For comparison, we also include  
the true likelihood function (TRUTH), which is unavailable in practical applications\footnote{Because the observed images are simulated, we can compute $a_i$ and $r_i$.}.
It is apparent from the figure that the spectral series estimator comes closer to the truth than 
the other estimators, even without first reducing the dimensionality of the galaxy images. 
In the  Supplementary Materials we present additional plots for other sample sizes. These results yield similar conclusions.
\begin{figure}[ht]
\vspace{.0in}
\centerline{\includegraphics[page=4,scale=0.34]{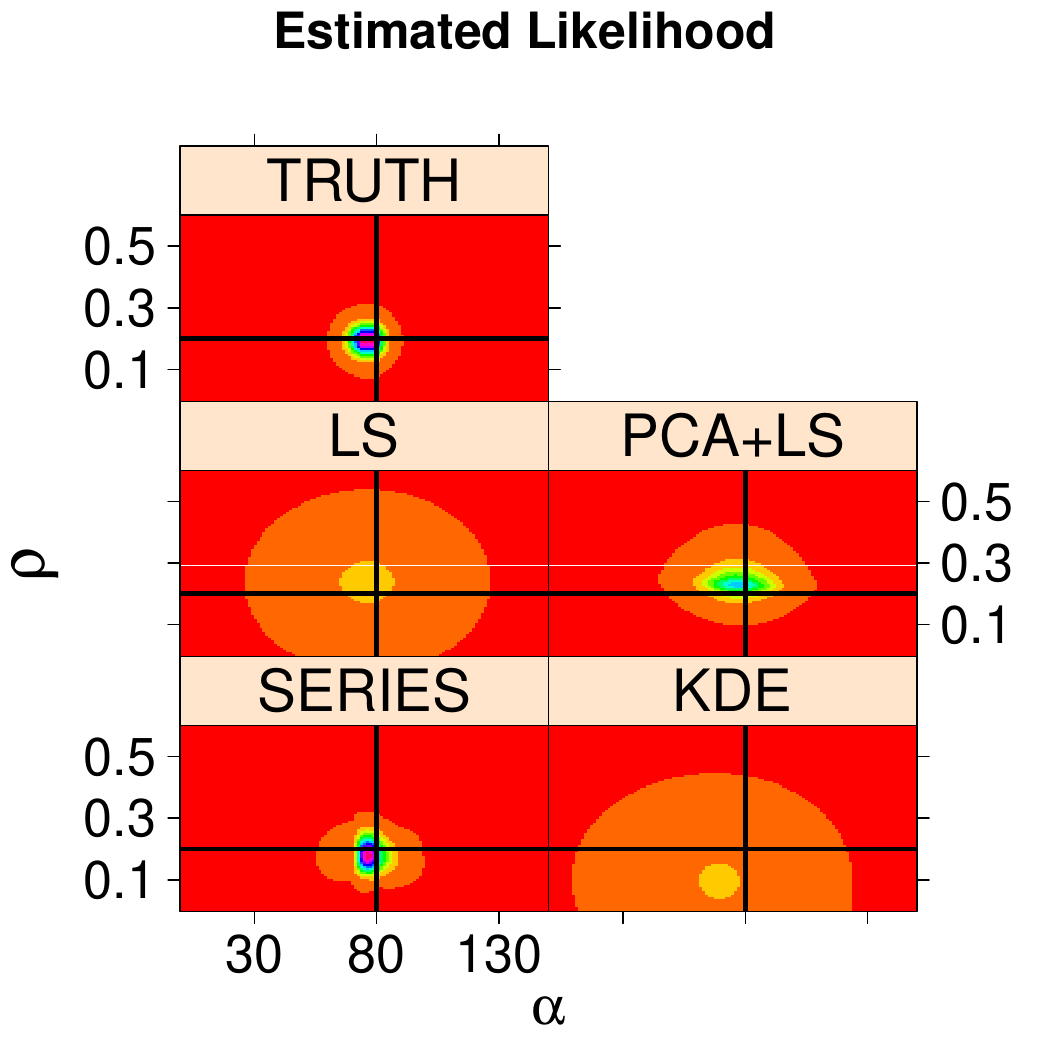}}
\vspace{.0in}
\caption{{\small Estimated likelihood function for the galaxy example using different estimators. The contours represent the level sets; 
 horizontal and vertical lines are the true values of the parameters. 
 The spectral series estimator gets close to the true distribution, which is uncomputable in practice.}}
 \label{fig::posteriors}
\end{figure}

\vspace{-1mm}
Furthermore, to quantify how the level sets of the likelihood function concentrate around the true parameters, we define the expected average
distance of the estimated likelihood function
to the real parameter value,  
$\E_{\vec{x},\theta^*}\left[\int d(\theta^*,\theta) \widehat{\L}(\vec{x};\theta)d\theta\right]$\footnote{As the prior distribution
is uniform, this quantity is $\E_{\vec{x},\theta^*}\left[\int d(\theta^*,\theta)d\widehat{f}(\theta|\vec{x})\right]$.}, where the expectation is taken with respect to both $\theta^*$ and the observed data. Here we choose  $d(\theta^*,\theta)$ to be the Euclidean distance between the  vectors of parameters, standardized so that each component
has minimum 0 and maximum 1. 
  As a final comparison of methods, we study how the above likelihood metric changes as a function of the sample size $m$ of the observed data (the sample size of the simulated data used to estimate the likelihood is held constant); see Figure~\ref{fig::distance} for results. Because $\L(\vec{x};\theta)$
concentrates around the true parameter value $\theta^*$ for large sample sizes, we expect the average likelihood  to decrease as $m$ increases -- if the likelihood estimates are reasonable.
Indeed, we observe this behavior for all methods in the comparison
for the problems with low dimensionality. 
However, for the problems with high dimensionality, this is no longer the case. 
On the other hand, the results indicate that \emph{Series} is able to overcome the curse of dimensionality 
 and recover the true $\theta^*$ parameter as the number of observations increases.

  \begin{figure}[ht]
 \vspace{.2in}
 \centerline{\includegraphics[page=1,scale=0.19]{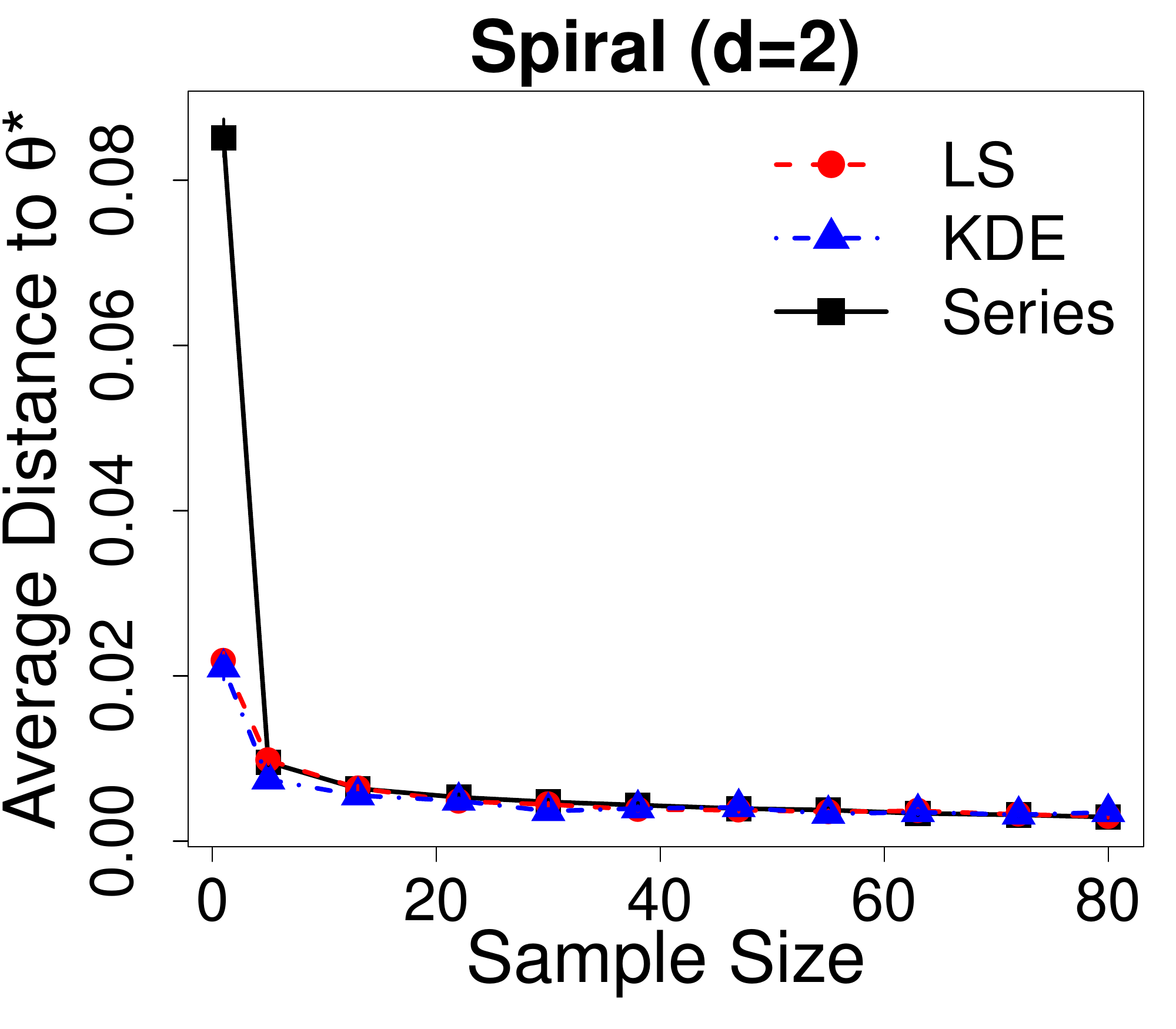}
 \includegraphics[page=1,scale=0.19]{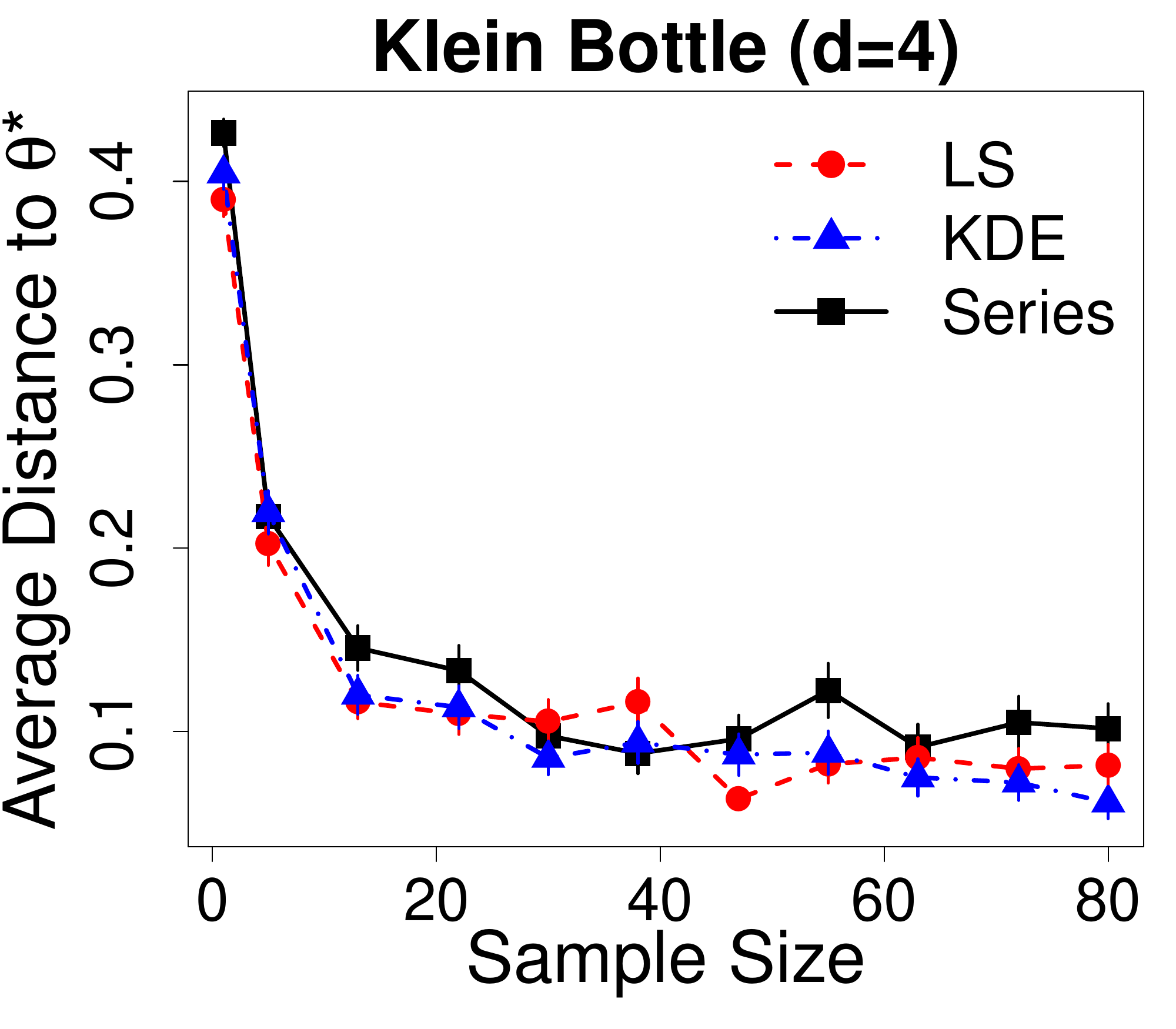}}
 \centerline{\includegraphics[page=1,scale=0.19]{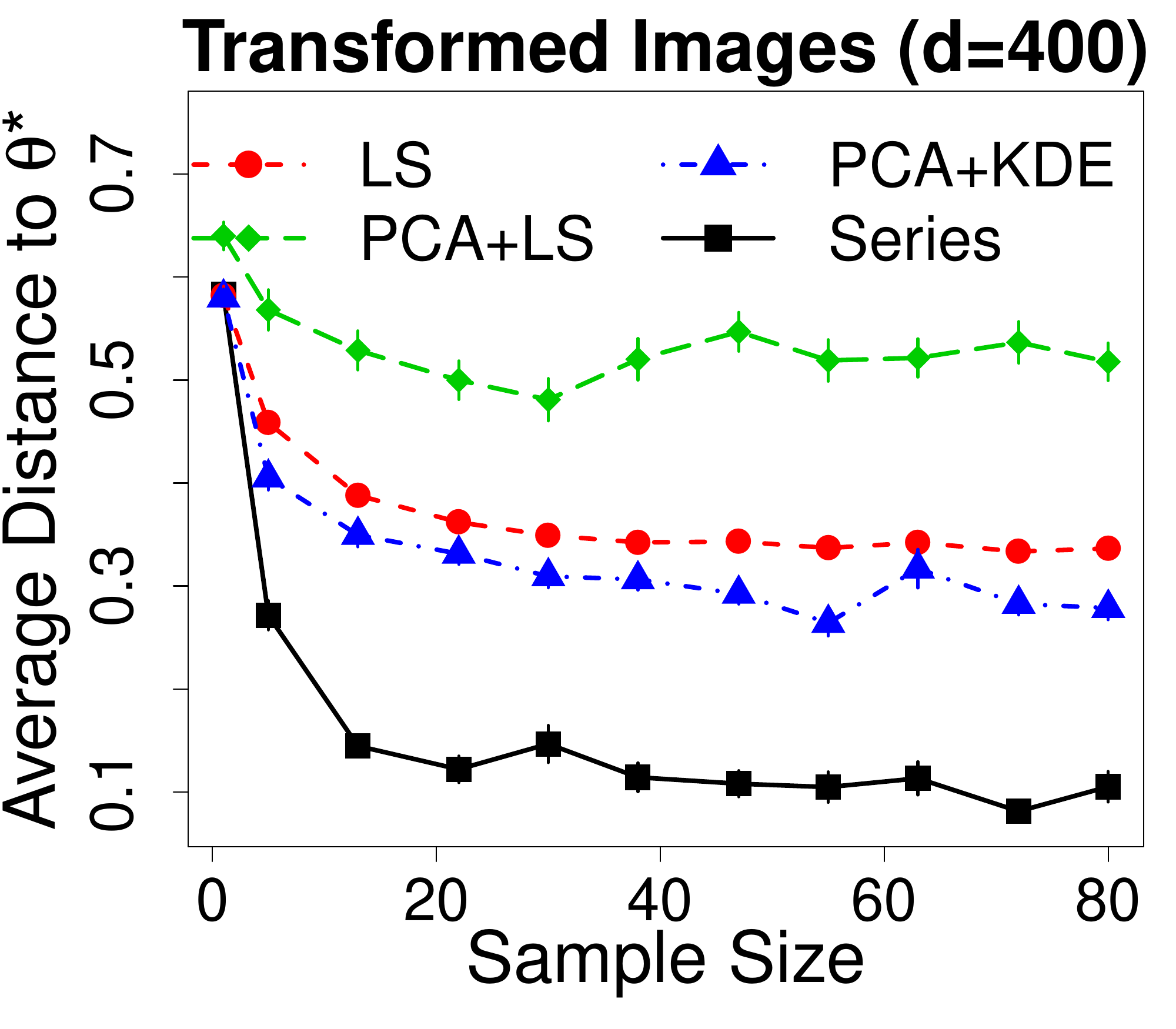}
 \includegraphics[page=1,scale=0.19]{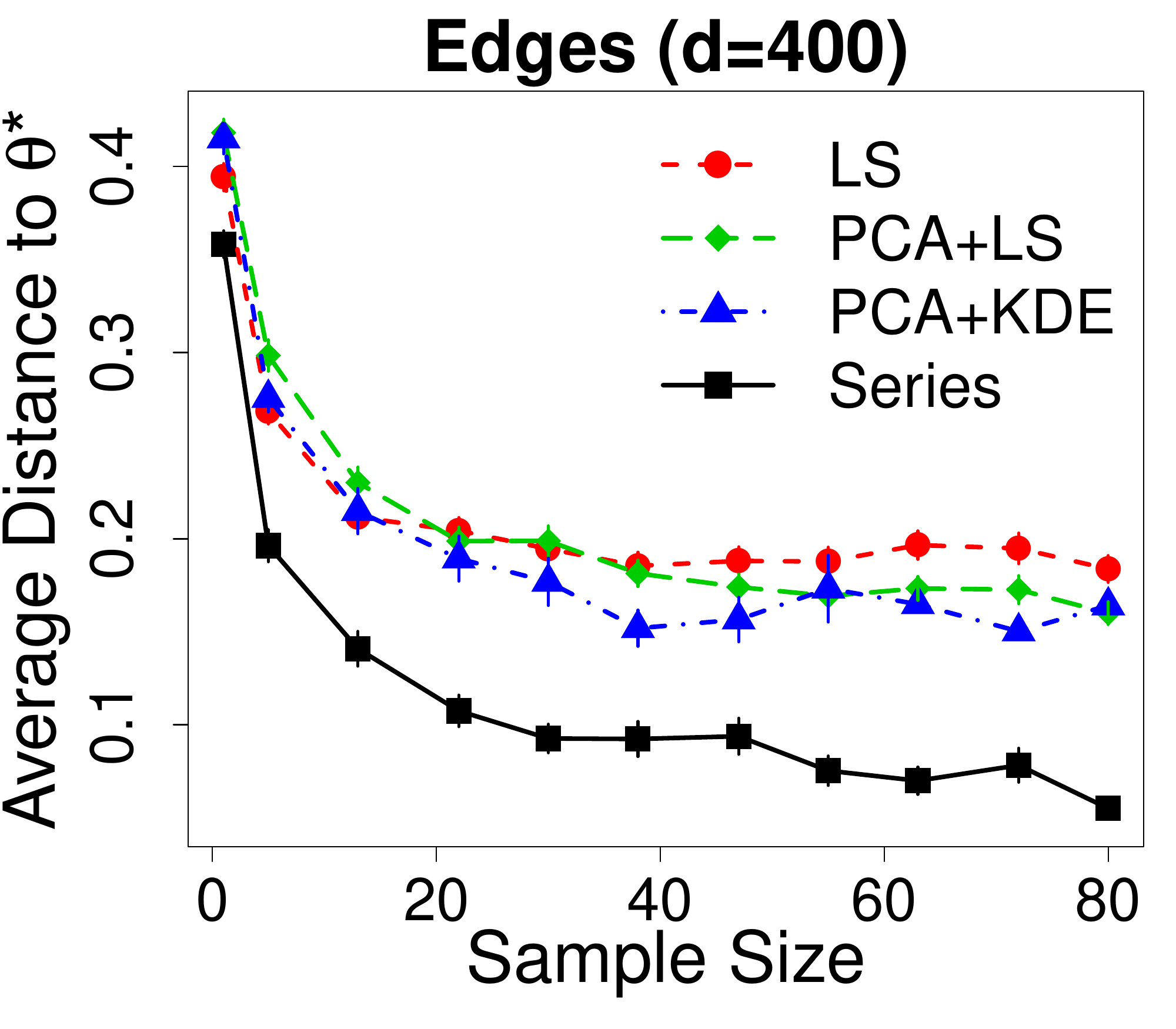}}
 \centerline{\includegraphics[page=1,scale=0.19]{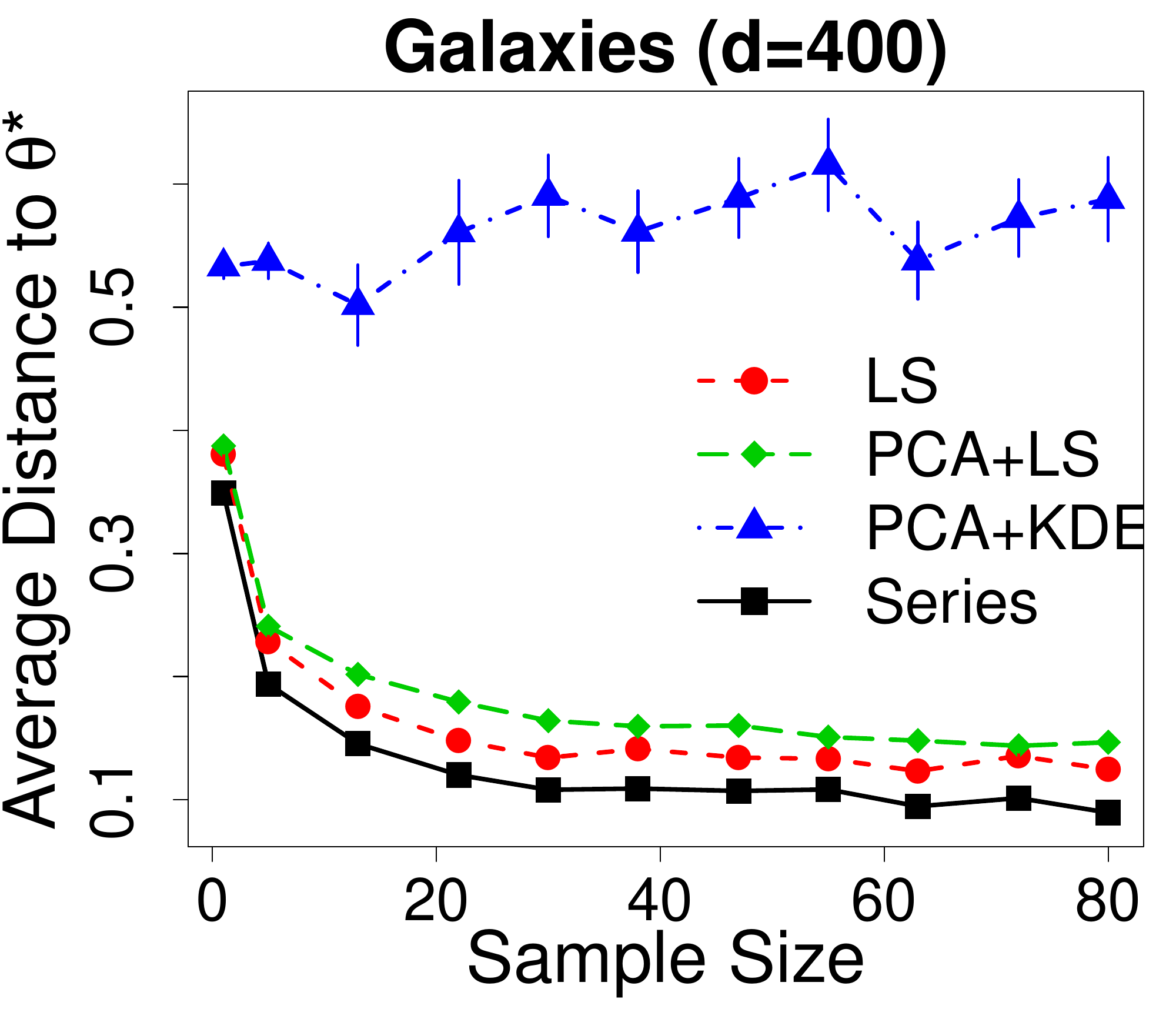}}
 \vspace{.0in}
 \caption{\small Average distance of estimated likelihoods to the true $\theta$  (and standard errors) as a function of the number of observed images for the galaxy data. While in low dimensions all estimators have similar performance, our approach performs better in high dimensions.}
  \label{fig::distance}
 \end{figure}
 
\section{THEORY}
\label{sec-theory}
Next we provide theoretical guarantees of the performance of the 
estimator $\widehat{\beta}$.
The integral operator
$\textbf{K}_\x$ from Eq. \ref{eq::kpcaOperator} is self-adjoint, compact and has a countable number of 
eigenfunctions $\{\psi_i\}_i$
with respective eigenvalues $\lambda_1 \geq \lambda_2 \geq \ldots \geq 0$.
These eigenfunctions therefore form an orthonormal basis of $L^2(\mathcal{X},G)$ \citep{Minh2}.
We make the following assumptions:

\vspace{1.2mm}
\begin{Assumption}
\label{assumpfiniteWeight}
$\int \beta^2(\vec{x})dG(\vec{x}) < \infty$.
\end{Assumption}

\vspace{1.2mm}
\begin{Assumption} 
 \label{decayEigenvaluesWeight}
$\lambda_1>\lambda_2>\ldots>\lambda_J>0$.
\end{Assumption}

Assumption \ref{assumpfiniteWeight} states that the ratio is $L^2$ integrable. It implies that it is possible to expand $\beta$ into
the basis $\psi$. 
Assumption \ref{decayEigenvaluesWeight} allows one to uniquely define each of the eigenfunctions
(see, e.g., \citealt{Yangicml12} for similar assumptions, and \citealt{Zwald} on how to proceed if it does not hold). 
Let $\mathcal{H}_{K_\x}$ denote the Reproducing Kernel Hilbert Space (RKHS) associated to the kernel $K_\x$.
We assume

\vspace{1.2mm}
\begin{Assumption} 
 \label{assump-hilbertXWeight} 
$c_{K_\x} \equiv ||\beta(\vec{x})||^2_{\mathcal{H}_{K_\x}} <\infty $.
\end{Assumption}

Assumption \ref{assump-hilbertXWeight} implies
smoothness of $\beta(\vec{x})$ as measured by the RKHS norm defined by $K_\x$. Smaller values of $c_{K_\x}$
imply smoother functions. In the Supplementary Materials we prove the following main result.

\vspace{2mm}
\begin{thm}
\label{thm-mainBoundWeights}
 Under Assumptions \ref{assumpfiniteWeight} -- \ref{assump-hilbertXWeight},
 the loss 
$\int \left(\widehat{\beta}_{J}(\vec{x})-\beta(\vec{x}) \right)^2 dG(\vec{x})$
 is bounded by
 \[
J \times \left[O_P\!\left(\frac{1}{n_F}\right)+ O_P\!\left(\frac{1}{\lambda_{J} \Delta^2_{J} n_G}\right) \right]+c_{K_\x}O\!\left( \lambda_{J}\right),
\]
 where $\Delta_{J}=\min_{1\leq j \leq J}\left|\lambda_j-\lambda_{j+1}\right|$
 and $\widehat{\beta}_{J}(\vec{x})$ is the spectral series 
 ratio estimator truncated at $J$.
\end{thm}
\vspace{2mm}

The first term of the rate in Theorem \ref{thm-mainBoundWeights} is the 
sample error. The second term is the approximation error. 
Smooth functions have a small value of $c_{K_\x}$, and therefore a smaller bias.
These rates depend on the decay of the eigenvalues $\lambda_J$ and the eigengaps
$\lambda_J-\lambda_{J+1}$.

As an illustration, assume a {\em fixed} kernel $K_\x$.
Then, if $n \equiv n_F=n_G$, $\lambda_J \asymp J^{-2\alpha}$ for some $\alpha>\frac{1}{2}$,
and $\lambda_J-\lambda_{J+1} \asymp J^{-2\alpha-1}$ (see \citealt{Yangicml12} for an empirical motivation), 
then the optimal smoothing is given by 
$J \asymp n^{1/(8\alpha+3) }.$ 
With this choice of $J$, the rate of convergence is
$$O_P\!\left(n^{-\frac{2\alpha}{8\alpha+3}}\right).$$
Note, however, that by changing the kernel (e.g., by using different bandwidths $\epsilon$),
one can improve the performance of the estimator: different kernels lead to different eigenvalue decays, as well as 
different $c_{K_\x}$'s (e.g., $\alpha$ depends on $\epsilon$). 
Hence, choosing the tuning parameters properly is important. 
Note that $J \times O_P\!\left(1/n_F\right)$ is the traditional variance of orthogonal classical estimators
\emph{in one dimension}, in which one expands the target function with 
respect to a basis \emph{fixed beforehand} (e.g., the Fourier basis) \citep{efromovich}. The additional term 
$J \times O_P\!\left(1/(\lambda_{J} \Delta^2_{J} n_G)\right)$ is the cost 
of estimating a basis (from data) that better captures the geometry of the data.

Under similar assumptions (see the Supplementary Materials for proofs and details), an analogous bound holds for the spectral series likelihood estimator truncated at $I$ and $J$:
\begin{align*}
&L\left(\widehat{\L}_{I,J},\L\right) = 
 IJ O_P\!\left(\max\left\{\frac{1}{\lambda^{\vec{x}}_J \Delta^2_{\vec{x},J} n_G},\frac{1}{\lambda^{\theta}_I \Delta^2_{\theta,I} n_F}\right\}\right) \\
&+c_{K_\theta}O\!\left(\lambda_{I}^\theta\right) +c_{K_\x}O\!\left(\lambda_{J}^\x\right), 
\end{align*}
where the superscript $\x$ and $\theta$ 
denote quantities associated with the eigenfunctions $\psi_j$'s and $\phi_i$'s, respectively. Similar interpretation holds for this bound.

\section{CONCLUSION}
\label{sec::remarks}

We have demonstrated the effectiveness of a new spectral series approach for estimating the ratio of two high-dimensional densities, with extensions to likelihood approximation in high dimensions. 
Traditional approaches typically fail for high-dimensional data (even with a dimension reduction by PCA) whereas the proposed method is designed to adapt to the underlying geometry of the data.

\comment{We showed that spectral series can be effectively used to estimate the ratio of two densities, as well as 
a likelihood function. Our approach is designed to work when the sample space has a small intrinsic dimensionality,
but may be embedded in a high-dimensional space, as is the 
case of images, spectra, tracks, and time series. In these cases, traditional
approaches typically fail. }

\section{Acknowledgments}

This work was partially supported by \emph{Conselho Nacional de Desenvolvimento Cient\'ifico e Tecnol\'ogico} (grant 200959/2010-7) 
and the Estella Loomis McCandless Professorship.



\bibliographystyle{plainnat}
\bibliography{bibCS}

\onecolumn{

\section*{Supplementary Material}

\textbf{Details on Photometric Redshift Problem and Sloan Digital Sky Survey Data}

In spectroscopy, 
the flux of a galaxy, i.e., the number of photons emitted per unit area per unit time, is measured as a function of wavelength.
By using these measurements it is possible to determine the redshift of a galaxy with great precision.
On the other hand, in photometry--an extremely low-resolution spectroscopy--
the photons 
are collected into a few ($\approx 5$) 
wavelength bins (also named bands). In each of these bins, the magnitudes -
which are logarithmic measurements of photon flux - are measured. 
Typical instruments measure in five bands, denoted by $u$, $g$, $r$, $i$, and $z$. 
The differences between contiguous magnitudes (also named \emph{colors}; e.g., $g-r$) are useful predictors for the
redshift of the galaxy.
Multiple estimators of the magnitudes exist, here we work with two of them: {\tt model} and {\tt cmodel} \citep{Sheldon}.
Our covariates are the 4 colors in each magnitude system, plus the raw value of $r$-band magnitude in both system.
Hence there are $4\times2=10$ covariates $\vec{x}$.

 Because it is difficult to acquire the spectroscopic redshift of faint galaxies, these data suffer from selection bias. 
 We take this into account by making the
 \emph{covariate shift} assumption: although $f_{L}(\vec{x})$ can be different 
from $f_{U}(\vec{x}),$ we assume the conditional distribution
$f(z|\vec{x})$ is the same in both populations \citep{Shimodaira2000}. 
This correction will
reweight labeled data to account for the difference in their distribution as compared to the 
unlabeled data \citep{Gretton}.

The data we use are similar to that used by \citet{Sheldon}. 
The labeled data set  contains spectroscopic information about 
435,875 galaxies from the Sloan Digital Sky Survey (SDSS). It also contains {\tt model} and {\tt cmodel} magnitudes in the bands $u$,
$g$, $r$, $i$ and $z$.
These samples were chosen by applying cuts
to both main sample galaxies  and luminous red galaxies of SDSS. 
Only galaxies in which the redshift
information has confidence level at least 0.9 were selected. 
It was also required that the galaxies were not too faint.
The unlabeled data set contains a subset of 538,974 galaxies of SDSS data. The only variables that are observed are the 
photometric magnitudes.
These samples were chosen by applying cuts to the original unlabeled SDSS data
imposing that the galaxies are not too faint
and have reasonable colors, see \citet{Sheldon} for more details.

\textbf{Additional Figures for Galaxy Likelihood Estimation Example}

 The first column of Figure \ref{fig-galaxyExamples} shows examples of galaxies with different parameter values,
 generated by GalSim toolkit.
 To make the situation more realistic, we assume we cannot observe the 
images of the uncontaminated galaxies in the first row, but instead 
only the 20 $\times$ 20 images from the last row. These are low-resolution
 images degraded by observational effects, background noise, and pixelization; see \citet{bridle2009handbook} for more details.

 \begin{figure}[h]
\vspace{.3in}
\centerline{\includegraphics[scale=1]{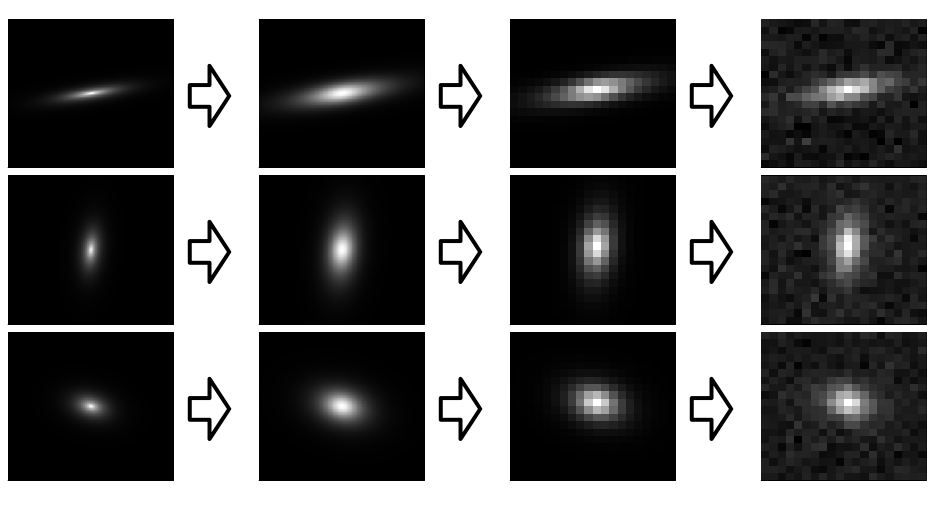}}
\vspace{.3in}
\caption{Examples of galaxies with different orientations and axis ratios. From left to right:
High-resolution, uncontaminated galaxy image; effect of PSF caused by atmosphere and telescope; pixelated image; and observed image containing additional Poisson noise. We only observe
images on the right.}
\label{fig-galaxyExamples}
\end{figure}

\begin{figure}[H]
  \begin{center}
 \includegraphics[page=4,scale=0.38]{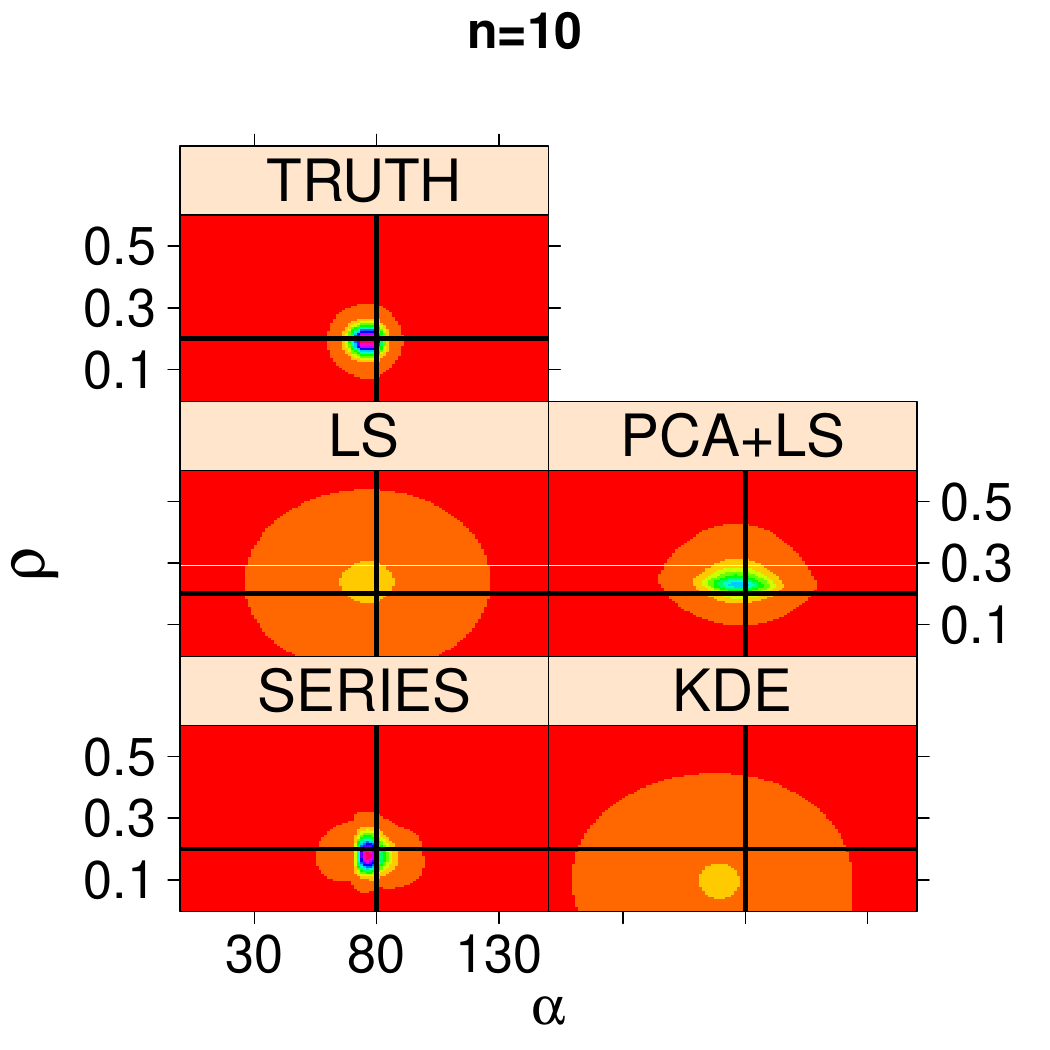}\hspace{-2mm}
 \includegraphics[page=4,scale=0.38]{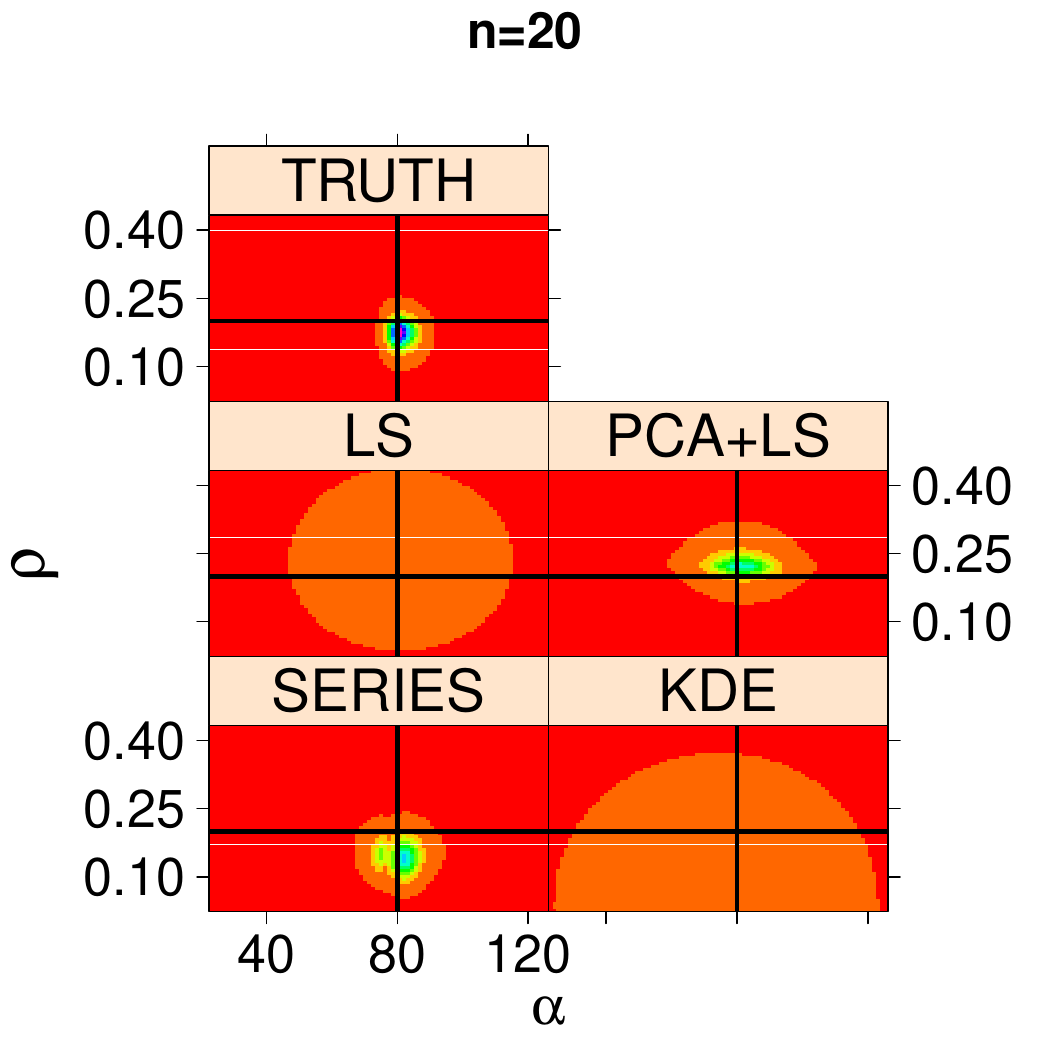}\\[-3mm]
 \includegraphics[page=4,scale=0.38]{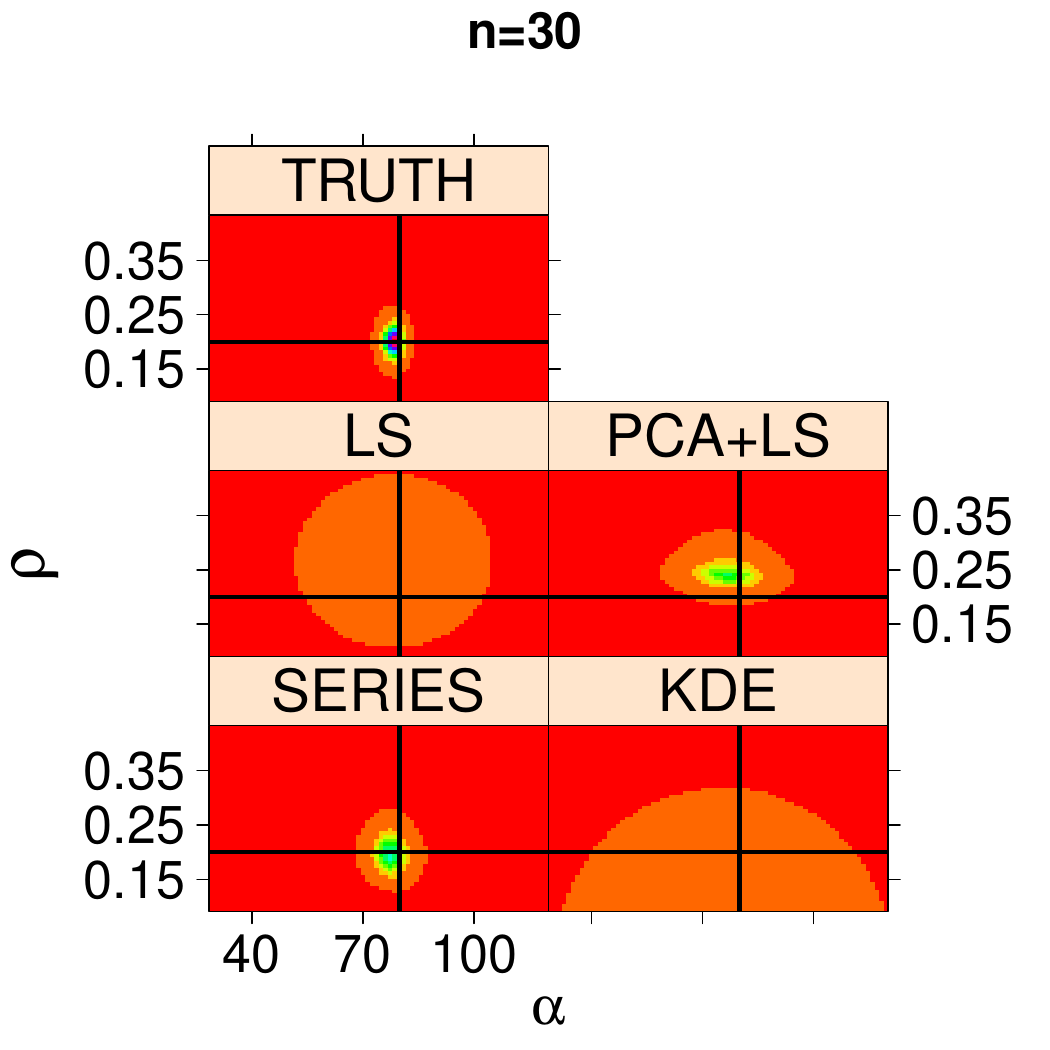}\hspace{-2mm}
 \includegraphics[page=4,scale=0.38]{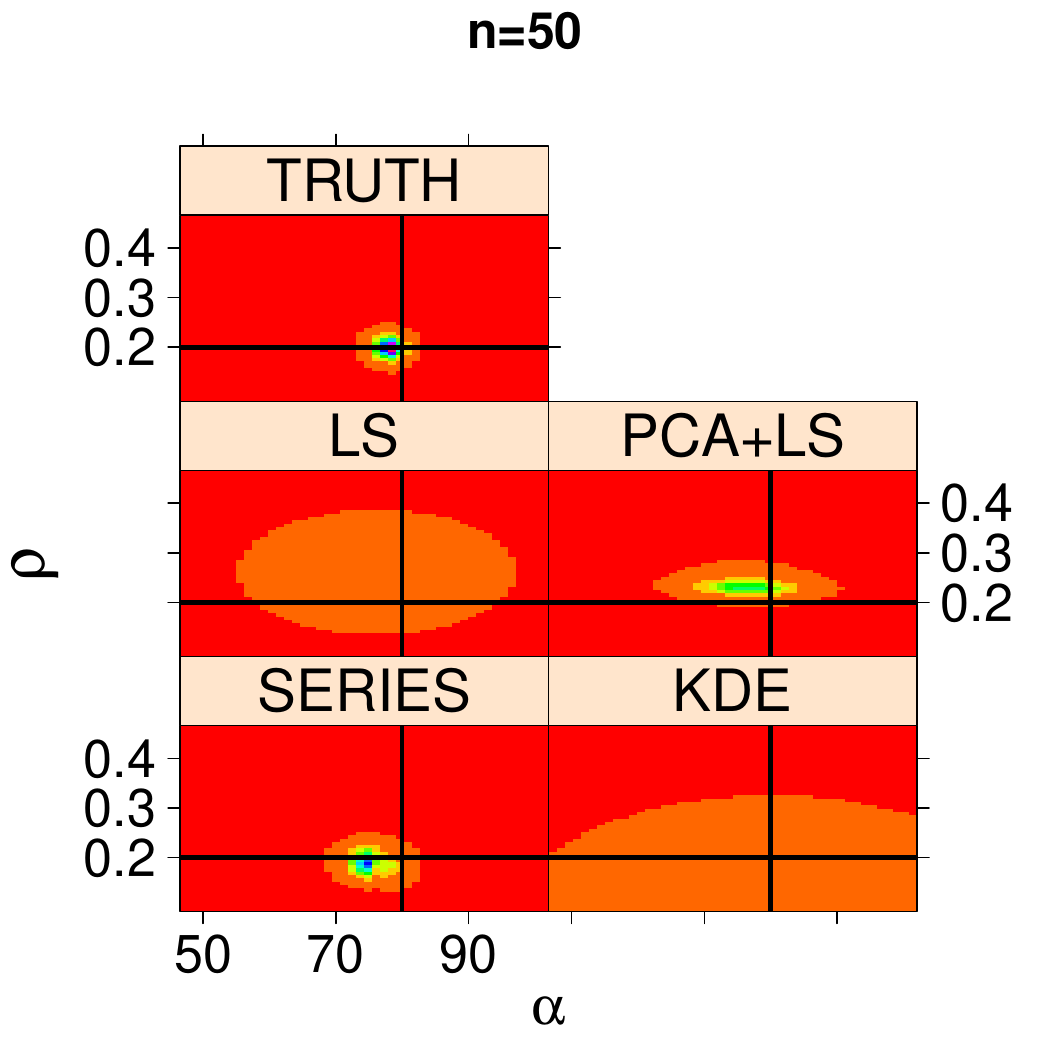}\\
  \end{center}
 \caption{Comparison of level sets of estimated likelihood function $\L(\x;(\alpha,\rho))$ for the galaxy example
 for 4 samples sizes. Horizontal and vertical lines are the true values of the parameters.
 In all cases, the spectral series estimator gets closer to the real distribution, which is uncomputable in practice.}
 \label{fig::posteriors2}
\end{figure}

\textbf{Spectral Series Density Ratio Estimator}

Here we prove the bounds from the paper.

Define the following quantities: 
\begin{align*}
&\beta_{J}(\vec{x})=\sum_{j=1}^J\beta_{j}\psi_j(\vec{x}),\quad \beta_{j}=\int \psi_j(\vec{x})dF(\vec{x})  \\ 
&\widehat{\beta}_{J}(\vec{x})=\sum_{j=1}^J\widehat{\beta}_{i}\widehat{\psi}_j(\vec{x}),\quad \widehat{\beta}_{j}=\frac{1}{n_F}\sum_{k=1}^n \widehat{\psi}_j(\vec{x}^F_k)
\end{align*}
and note that
\begin{eqnarray*}
 \int \left(\widehat{\beta}_{J}(\vec{x})-\beta(\vec{x}) \right)^2 dG(\vec{x}) 
 &\leq &  \int \left(\widehat{\beta}_{J}(\vec{x})-\beta_{J}(\vec{x})+\beta_{J}(\vec{x})- \beta(\vec{x})\right)^2 dG(\vec{x})  \\
& \leq & 2\left(\mbox{VAR}(\widehat{\beta}_{J}(\x),\beta_{J}(\x))+B(\beta_{J}(\x),\beta(\x))\right).
\end{eqnarray*}
where 
$$B(\beta(\x)_{J},\beta(\x))\equiv \iint \left(\beta_{J}(\vec{x})-\beta(\vec{x}) \right)^2 dG(\vec{x})$$ can be interpreted as a bias 
term (or approximation error) and $$\mbox{VAR}(\widehat{\beta}_{J}(\x),\beta_{J}(\x)) \equiv \iint \left(\widehat{\beta}_{J}(\vec{x})-\beta_{J}(\vec{x}) \right)^2 dG(\vec{x})dz$$
can be interpreted as a variance term. First we bound the variance.

\begin{Lemma} 
 \label{boundBeta} 
There exists $C>0$ such that
$|\beta(\vec{x})| < C$ for all $\vec{x}\in \mathcal{X}$.

\end{Lemma}

\begin{proof}
 Using Assumption \ref{assump-hilbertXWeight} and the fact the kernel is bounded, it follows from the reproducing property 
 and Cauchy-Schwartz inequality that 
 \begin{align*}
\beta(\vec{x}) = \langle\beta(.),K(\vec{x},.) \rangle_{\mathcal{H}_{K_\x}} \leq || \beta(.)||_{\mathcal{H}_{K_\x}}\sqrt{K_\x(\vec{x},\vec{x})} <C 
 \end{align*}
 for some $C>0$.
\end{proof}

\begin{Lemma} 
 \label{boundEigenVectors} For all $1\leq j \leq J$,
$$\int \left(\widehat{\psi_j}(\vec{x})-\psi_j(\vec{x}) \right)^2dG(\vec{x}) = O_P\left(\frac{1}{\lambda_j \delta^2_j n_G}\right),$$
where $\delta_j=\lambda_j-\lambda_{j+1}$. 
\end{Lemma}

For a proof of Lemma \ref{boundEigenVectors} see, e.g., \citet{Sinha}.

\begin{Lemma} 
 \label{boundEigenVectorsF} For all $1\leq j \leq J$,
$$\int \left(\widehat{\psi_j}(\vec{x})-\psi_j(\vec{x}) \right)^2dF(\vec{x}) = O_P\left(\frac{1}{\lambda_j \delta^2_j n_G}\right),$$
where $\delta_j=\lambda_j-\lambda_{j+1}$. 
\end{Lemma}

\begin{proof}
 It follows  from Lemmas \ref{boundBeta} and \ref{boundEigenVectors} that 
 \begin{align*}
 &\int \left(\widehat{\psi_j}(\vec{x})-\psi_j(\vec{x}) \right)^2dF(\vec{x}) =
 \int \left(\widehat{\psi_j}(\vec{x})-\psi_j(\vec{x}) \right)^2\beta(\vec{x})dG(\vec{x}) \leq \\
 &C\int \left(\widehat{\psi_j}(\vec{x})-\psi_j(\vec{x}) \right)^2dG(\vec{x}) =
 O_P\left(\frac{1}{\lambda_j \delta^2_j n_G}\right)
 \end{align*}

\end{proof}

\begin{Lemma} \label{lemmaboundExp2} For all $1\leq j \leq J$, there exists $C < \infty$ that does not depend on $n_G$ such that
 $$E\left[\left(\widehat{\psi}_j(\vec{X}^F)-\psi_j(\vec{X}^F)\right)^2\right] < C.$$
\end{Lemma}
\begin{proof}
Let $\delta\in(0,1)$. From \citet{Sinha}, it follows that
\begin{align*}
 \P\left(\int\left(\widehat{\psi}_j(\vec{x})-\psi_j(\vec{x})\right)^2dG(\vec{x})> \frac{16\log\left(\frac{2}{\delta}\right)}{\delta^2_j n_G}\right) < \delta,
 \end{align*}
 and therefore for all $\epsilon>0$,
 \begin{align*}
 \P\left(\int\left(\widehat{\psi}_j(\vec{x})-\psi_j(\vec{x})\right)^2dG(\vec{x})> \epsilon \right) < 2e^{-\frac{\delta_j^2 n_G\epsilon}{16}}.
 \end{align*}
 Hence , using Lemma \ref{boundBeta},
 \begin{align*}
 &\E\left[\left(\widehat{\psi}_j(\vec{X}^F)-\psi_j(\vec{X}^F)\right)^2\right] =\E\left[\int\left(\widehat{\psi}_j(\vec{x})-\psi_j(\vec{x})\right)^2dF(\vec{x})\right] \leq C\E\left[\int\left(\widehat{\psi}_j(\vec{x})-\psi_j(\vec{x})\right)^2dG(\vec{x})\right] \\
 &\int_{0}^\infty \P\left(\int\left(\widehat{\psi}_j(\vec{x})-\psi_j(\vec{x})\right)^2dG(\vec{x})> \epsilon \right) d\epsilon \leq \int 2e^{-\frac{\delta_j^2 n_G\epsilon}{16}} d\epsilon < \int 2e^{-\frac{\delta_j^2 \epsilon}{16}} d\epsilon < \infty
 \end{align*}
\end{proof}
\vspace{2mm}

\begin{Lemma}
 \label{boundVariance} For all $1\leq j \leq J$, there exists $C < \infty$ that does not depend on $m$ such that
 $$\E\left[ \V\left[\left(\widehat{\psi}_j(\vec{X}^F)-\psi_j(\vec{X}^F)\right) \bigg\vert \vec{X}^G_1,\ldots,\vec{X}^G_{n_G} \right] \right]< C$$ 
\end{Lemma}

\begin{proof}
We have that
\begin{align*}
 &\E\left[\V\left[\left(\widehat{\psi}_j(\vec{X}^F)-\psi_j(\vec{X}^F)\right) \bigg\vert \vec{X}^G_1,\ldots,\vec{X}^G_{n_G} \right]\right]\\
 & \leq \V\left[ \widehat{\psi}_j(\vec{X}^F)-\psi_j(\vec{X}^F) \right]  \leq  \E\left[\left(\widehat{\psi}_j(\vec{X}^F)-\psi_j(\vec{X}^F)\right)^2\right]
\end{align*}
The result follows from Lemma \ref{lemmaboundExp2}.
\end{proof}
\vspace{2mm}

\begin{Lemma} 
\label{boundEigenVectorsSample} For all $1\leq j \leq J$,
$$ \left[ \frac{1}{n}\sum_{k=1}^n\left(\widehat{\psi}_j(\vec{X}^F_k)-\psi_j(\vec{X}^F_k)\right)-\int\left(\widehat{\psi}_j(\vec{x})-\psi_j(\vec{x})\right)dF(\vec{x})\right]^2=O_P\left(\frac{1}{n}\right)$$
\end{Lemma}
\begin{proof}
 By Chebyshev's inequality it holds that for all $M>0$ 
 \begin{align*}
&\P\left(\left|\frac{1}{n}\sum_{k=1}^n\left(\widehat{\psi}_j(\vec{X}^F_k)-\psi_j(\vec{X}^F_k)\right)-\int\left(\widehat{\psi}_j(\vec{x})-\psi_j(\vec{x})\right)dF(\vec{x}) \right|^2>M \bigg\vert \vec{X}^G_1,\ldots,\vec{X}^G_{n_G} \right) \leq \\
&\frac{1}{n_FM} \V\left[\left(\widehat{\psi}_j(\vec{X}^F)-\psi_j(\vec{X}^F)\right) \bigg\vert \vec{X}^G_1,\ldots,\vec{X}^G_{n_G} \right]. 
\end{align*}
The conclusion follows from taking an expectations with respect to sample from $G$ on both sides of the equation and using Lemma \ref{boundVariance}.

\end{proof}
\vspace{2mm}

Note that $\widehat{\psi}'$s are random functions, and therefore the proof of Lemma \ref{boundEigenVectorsSample}
relies on the fact that these functions are estimated using a different sample than $\vec{X}_1,\ldots,\vec{X}_n$.

\begin{Lemma} 
\label{boundBetaHat} For all $1\leq j \leq J$,
$$\left(\widehat{\beta}_{j}-\beta_{j}\right)^2 = O_P\left(\frac{1}{n}\right)+ O_P\left(\frac{1}{\lambda_j \delta^2_jn_G}\right).$$
\end{Lemma}
\begin{proof} It holds that
 \begin{align*}
\frac{1}{2}\left(\widehat{\beta}_{j}-\beta_{j}\right)^2 &\leq \left( \frac{1}{n_F}\sum_{k=1}^{n_F}\psi_j(\vec{X}^F_k)-\beta_{j}\right)^2+ \left( \frac{1}{n_F}\sum_{k=1}^{n_F}(\widehat{\psi}_j(\vec{X}^F_k)-\psi_j(\vec{X}^F_k))\right)^2.
 \end{align*}
 The first term is $O_P\left(\frac{1}{n_F}\right)$. 
 The second term divided by two is bounded by
  \begin{align*}
&\frac{1}{2}\left(  \frac{1}{n_F}\sum_{k=1}^{n_F}(\widehat{\psi}_j(\vec{X}^F_k)-\psi_j(\vec{X}^F_k))-\int (\widehat{\psi}_j(\vec{x})-\psi_j(\vec{x}))dF(\vec{x})+\int (\widehat{\psi}_j(\vec{x})-\psi_j(\vec{x}))dF(\vec{x}) \right)^2\\
&\leq \left( \frac{1}{n_F}\sum_{k=1}^{n_F}(\widehat{\psi}_j(\vec{X}^F_k)-\psi_j(\vec{X}^F_k))-\int (\widehat{\psi}_j(\vec{x})-\psi_j(\vec{x}))dF(\vec{x})\right)^2+\int (\widehat{\psi}_j(\vec{x})-\psi_j(\vec{x}))^2dF(\vec{x})\\
&=O_P\left(\frac{1}{n_F}\right)+O_P\left(\frac{1}{\lambda_j \delta^2_jn_G}\right)
 \end{align*}
The result follows from Lemma \ref{boundEigenVectorsF}.
 
\end{proof}
\vspace{2mm}

\begin{Lemma}\textbf{[\citealt{Sinha}, Corollary 1]}
\label{boundPsiHat} 
Under the stated assumptions,
$$ \int \widehat{\psi}_j^2(\vec{x})dG(\vec{x}) = O_P\left(\frac{1}{\lambda_j \Delta^2_J n_G}\right)+1$$
and
$$ \int \widehat{\psi}_i(\vec{x}) \widehat{\psi}_j(\vec{x})dG(\vec{x}) = O_P\left(\left(\frac{1}{\sqrt{\lambda_i}}+\frac{1}{\sqrt{\lambda_j}}\right) \frac{1}{\Delta_J \sqrt{n_G}}\right)$$
where $\Delta_J=\min_{1\leq j \leq J}\delta_j$.
\end{Lemma} 

\begin{Lemma} 
\label{lemma-1stTerm} 
Let $h(\vec{x})=\sum_{j=1}^J \beta_{j}\widehat{\psi}_j(\vec{x}).$ Then

$$\int \left|\widehat{\beta}_{J}(\vec{x})-h(\vec{x}) \right|^2 dG(\vec{x}) = J \left(O_P\left(\frac{1}{n_F}\right)+ O_P\left(\frac{1}{\lambda_J \Delta^2_J n_G}\right) \right).$$ 
\end{Lemma} 
\begin{proof} 

\begin{align*}
&\int \left|\widehat{\beta}_{J}(z|\vec{x})-h(\vec{x})\right|^2 dG(\vec{x}) \\
&=\sum_{j=1}^J \left(\widehat{\beta}_{j}-\beta_{j}\right)^2 \int \widehat{\psi}_j^2(\vec{x})dG(\vec{x})+\sum_{j=1}^J \sum_{l=1, l \neq j}^J \left(\widehat{\beta}_{j}-\beta_{j}\right)\left(\widehat{\beta}_{l}-\beta_{l}\right) \int \widehat{\psi}_j(\vec{x}) \widehat{\psi_l}(\vec{x})dG(\vec{x}) \\
& \sum_{j=1}^J \left(\widehat{\beta}_{j}-\beta_{j}\right)^2 \int \widehat{\psi}_j^2(\vec{x})dG(\vec{x})+\left[ \sum_{j=1}^J  \left(\widehat{\beta}_{j}-\beta_{j}\right)^2 \right] \left[ \sqrt{\sum_{j=1}^J \sum_{l=1, l \neq j}^J \left( \int \widehat{\psi}_j(\vec{x}) \widehat{\psi}_l(\vec{x})dG(\vec{x}) \right)^2 } \right]
\end{align*}

where the last inequality follows from  using Cauchy-Schwartz repeatedly. The result follows from Lemmas \ref{boundBetaHat} and \ref{boundPsiHat}.

\end{proof}
\vspace{2mm}

\begin{Lemma} 
\label{lemma-2stTerm} 
Let $h(\vec{x})$ be as in Lemma \ref{lemma-1stTerm}. Then

$$\int \left|h(\vec{x})-\beta_{J}(\vec{x}) \right|^2 dG(\vec{x}) = J O_P\left(\frac{1}{\lambda_J \Delta^2_J n_G}\right).$$ 
\end{Lemma} 
\begin{proof} Using Cauchy-Schwartz inequality,
\begin{align*}
&\int \left| h(\vec{x})-\beta_{J}(\vec{x}) \right|^2 dG(\vec{x}) \leq \int \left|\sum_{j=1}^J \beta_{j}\left(\psi_j(\vec{x)}-\widehat{\psi_j}(\vec{x})\right)\right|^2 dG(\vec{x})\\
&=\left\{ \sum_{j=1}^J   \beta^2_{j}  \right\} \left\{ \sum_{j=1}^J \int \left[  \psi_j(\vec{x})-\widehat{\psi_j}(\vec{x})\right]^2 dG(\vec{x}) \right\}.
\end{align*}
The conclusion follows from Lemma \ref{boundEigenVectors} and by noticing that $\sum_{j=1}^J   \beta^2_{j} \leq ||\beta(\vec{x})||^2<\infty$.

\end{proof}
\vspace{2mm}

It is now possible to bound the variance term:

\begin{thm} 
\label{thm-var}
Under the stated assumptions,
$$\mbox{VAR}(\widehat{\beta}_{J}(\x),\beta_{J}(\x))=J \left(O_P\left(\frac{1}{n_F}\right)+ O_P\left(\frac{1}{\lambda_J \Delta^2_J n_G}\right) \right).$$
\end{thm}
\begin{proof} 
Let $h$ be defined as in Lemma \ref{lemma-1stTerm}. We have
\begin{align*}
&\frac{1}{2}\mbox{VAR}(\widehat{\beta}_{J}(\x),\beta_{J}(\x))= \frac{1}{2}\int \left|\widehat{\beta}_{J}(\vec{x})-h(\vec{x})+h(\vec{x})-\beta_{J}(\vec{x}) \right|^2 dG(\vec{x})  \\
&\leq \int \left|\widehat{\beta}_{J}(\vec{x})-h(\vec{x})\right|^2 dG(\vec{x}) + \int \left|h(\vec{x})-\beta_{J}(\vec{x}) \right|^2 dG(\vec{x}).
\end{align*}
The conclusion follows from Lemmas \ref{lemma-1stTerm} and \ref{lemma-2stTerm}.

\end{proof}

We now bound the bias term.

\begin{Lemma} 
\label{sobolevX}
$\sum_{j\geq J}\beta_j^2 = c_{K_\x} O(\lambda_J)$.

\end{Lemma}
\begin{proof}
Note that $c_{K_\x}=||\beta(\vec{x})||^2_{\mathcal{H}} = \sum_{j \geq 1} \frac{\beta_j^2}{\lambda_j}$ \citep{Minh2}.
Using Assumption \ref{assump-hilbertXWeight} and that  the eigenvalues are decreasing it follows that
\begin{align*}
\sum_{j \geq J} \beta_j^2=  \sum_{j \geq J} \beta_j^2 \frac{\lambda_j}{\lambda_j} \leq \lambda_J ||\beta(\vec{x})||^2_{\mathcal{H}},
\end{align*}
and therefore $\sum_{j\geq J}\beta_j^2 \leq  \lambda_J   c_{K_\x} = c_{K_\x} O(\lambda_J)$.

\end{proof}
\vspace{2mm}

\begin{thm}
\label{thm-bias}
Under the stated assumptions, the bias is bounded by
 $$B(\beta_{J}(\x),\beta(\x)) = c_{K_\x} O(\lambda_J).$$
 
\end{thm}
\begin{proof}
By using orthogonality, we have that
\begin{align*}
 B(\beta_{J}(\x),\beta(\x)) &\overset{\mbox{\tiny{def}}}{=} \int \left(\beta(\vec{x})-\beta_{J}(\vec{x}) \right)^2 dG(\vec{x}) = \sum_{j>J}\beta^2_{j}.
\end{align*}
The Theorem follows from Lemma
\ref{sobolevX}.
\end{proof}
\vspace{2mm}

\textbf{Spectral Series Likelihood Estimator}

We now present similar bounds to those from shown for the Density Ratio estimator. 
To avoid confusions with the last section, from now on we denote by $\lambda^{\vec{x}}_j$
the eigenvalue $\lambda_j$ relative to the eigenfunction $\psi_j$, and by $\delta^\vec{x}_j$
its eigengap previously denoted by $\delta_j$.

In this section, we assume Assumption \ref{assumpfiniteWeight} and \ref{decayEigenvaluesWeight} from
the last section, and, additionally:

\begin{Assumption} 
 \label{assump-hilbertX} For all $\vec{\theta} \in \Theta,$
 let 
 $g_\theta: \mathcal{X} \longrightarrow \Re;$ $g_\theta(\vec{x})=\L(\vec{x};\theta)$. 
$g_\theta \in \mathcal{H}_{K_\x} (c_\theta) \equiv \{g \in \mathcal{H}_{K_\x}: ||g||^2_{\mathcal{H}_{K_\x}} \leq c_\theta^2 \}$  where $c_\theta$'s are such that
$c_{K_\x}\equiv \int_{\Theta} c_\theta^2 dF(\theta)  < \infty $.
\end{Assumption}
\vspace{1mm}

Assumption \ref{assump-hilbertX} requires that for every $\theta \in \Theta$ fixed, $\L(\vec{x};\theta)$
is a smooth function of $\x$, where we again measure smoothness in a RKHS through its norm, and is analogous to Assumption \ref{assump-hilbertXWeight}. The last assumption 
we need is analogous to \ref{assump-hilbertX}, and requires that for every $\x \in \mathcal{X}$ fixed, $\L(\vec{x};\theta)$
is a smooth function of $\theta$:

\begin{Assumption} 
 \label{assump-hilbertz} For all $\vec{x} \in \mathcal{X},$
 let 
 $h_\x: \Theta \longrightarrow \Re;$ $h_\x(\vec{\theta})=\L(\vec{x};\theta)$. 
$h_\x \in \mathcal{H}_{K_\theta} (c_\x)\equiv \{h \in \mathcal{H}_{K_\theta}: ||h||^2_{\mathcal{H}_{K_\theta}} \leq c_\x^2 \}$  where $c_\x$'s are such that
$c_{K_\theta}\equiv  \int_{\mathcal{X}} c_\x^2 dG(\x)  < \infty $.
\end{Assumption}
\vspace{1mm}

First we note that an analogous decomposition of the loss 
$$L\left(\widehat{\L},\L\right) =\int \left(\widehat{\L}(\vec{x};\theta)-\L(\vec{x};\theta) \right)^2dG(\vec{x})dF(\theta)$$
in terms of bias and variance holds. The proof that the bound on the variance term is analogous to the one of the 
variance bound of the density ratio estimator. The main difference is in Lemmas \ref{boundEigenVectors} and \ref{boundPsiHat}.
We state and prove the new version of these in Lemmas \ref{boundEigenVectorsRatio} and \ref{boundPsiHatRatio}.
Notice that analogous Lemmas to \ref{boundEigenVectors} and \ref{boundPsiHat} hold for basis $\phi_i$.

\begin{Lemma} 
 \label{boundEigenVectorsRatio} For all $1\leq i \leq I$ and for all $1\leq j \leq J$,
$$\iint \left(\widehat{\Psi_{i,j}}(\vec{\theta},\vec{x})-\Psi_{i,j}(\vec{\theta},\vec{x}) \right)^2dG(\vec{x})dF(\vec{\theta}) = O_P\left(\max\left\{\frac{1}{\lambda^{\vec{x}}_j \delta^2_{\vec{x},j} n_G},\frac{1}{\lambda^{\vec{\theta}}_i \delta^2_{\vec{\theta},i} n_F}\right\}\right)$$,
\end{Lemma}

\begin{proof} We have that 
\begin{align*}
&\frac{1}{2}\iint \left(\widehat{\Psi_{i,j}}(\vec{\theta},\vec{x})-\Psi_{i,j}(\vec{\theta},\vec{x}) \right)^2dG(\vec{x})dF(\vec{\theta}) \leq  \\
 &\iint \left( (\widehat{\phi}_{i}(\vec{\theta})-\phi_{i}(\vec{\theta}))\widehat{\psi}_j(\vec{x})+\phi_i(\vec{\theta})(\widehat{\psi}_j(\vec{x})-\psi_j(\vec{x})) \right)^2dG(\vec{x})dF(\vec{\theta}) \leq \\
 &\int \widehat{\psi}_j^2(\vec{x}) dG(\vec{x})\int (\widehat{\phi}_{i}(\vec{\theta})-\phi_{i}(\vec{\theta}))^2 dF(\vec{\theta})+\int \phi^2_i(\vec{\theta})dF(\vec{\theta})\int (\widehat{\psi}_j(\vec{x})-\psi_j(\vec{x}))^2 dG(\vec{x})
\end{align*}
The results follows from orthonormality of $\phi_i$ wrt to $F(\vec{\theta})$, Lemmas \ref{boundEigenVectors} and \ref{boundPsiHat}.
\end{proof}
\vspace{2mm}

\begin{Lemma}
\label{boundPsiHatRatio} 
Under the stated assumptions, $\iint \widehat{\Psi}_{i,j}(\theta,\vec{x})\widehat{\Psi}_{k,l}(\theta,\vec{x})dG(\vec{x})dF(\theta) =$
\[ =
\left\{ 
  \begin{array}{l l}
1+O_P\left(\max\left\{\frac{1}{\lambda^{\vec{x}}_j \delta^2_{\vec{x},j} n_G},\frac{1}{\lambda^{\vec{\theta}}_i \delta^2_{\vec{\theta},i} n_F}\right\}\right), & \text{if $i=k$ and $j=l$}\\
O_P\left(\left(\frac{1}{\sqrt{\lambda^\x_l}}+\frac{1}{\sqrt{\lambda^\x_j}}\right) \frac{1}{\Delta^\x_J \sqrt{n_G}}\right) & \text{if $i= k$ and $j\neq l$}\\
O_P\left(\left(\frac{1}{\sqrt{\lambda^\theta_i}}+\frac{1}{\sqrt{\lambda^\theta_k}}\right) \frac{1}{\Delta^\theta_I \sqrt{n_F}}\right) & \text{if $i\neq k$ and $j= l$}\\
O_P\left(\max \left\{ \left(\frac{1}{\sqrt{\lambda^\x_l}}+\frac{1}{\sqrt{\lambda^\x_j}}\right) \frac{1}{\Delta^\x_J \sqrt{n_G}},\left(\frac{1}{\sqrt{\lambda^\theta_i}}+\frac{1}{\sqrt{\lambda^\theta_k}}\right) \frac{1}{\Delta^\theta_I \sqrt{n_F}} \right\} \right) & \text{if $i\neq k$ and $j\neq l$}\\
\end{array} \right.\]

\begin{proof} The proof of these facts follow from noticing that $\iint \widehat{\Psi}_{i,j}(\theta,\vec{x})\widehat{\Psi}_{k,l}(\theta,\vec{x})dG(\vec{x})dF(\theta) =
\int \widehat{\psi}_{j}(\vec{x})\widehat{\psi}_{l}(\vec{x})dG(\vec{x}) \int \phi_i(\theta) \phi_k(\theta)dF(\theta) $ and using Lemma \ref{boundPsiHat}
and its analogous for the basis $\phi_i$.

\end{proof}
\vspace{2mm}

\end{Lemma} 

The bound on the bias presents some additional differences to the proof of the bias bound from the ratio estimator, we therefore
show it in details in the sequence.

\begin{Lemma}
 \label{alphaExpansionZ}
For each $\theta\in\Theta,$ expand $g_\theta(\vec{x})$ into the basis $\psi:$ $g_\theta(\vec{x})=\sum_{j\geq1}\alpha_j^z\psi_j(\vec{x}),$
where $\alpha_j^\theta=\int g_\theta(\vec{x})\psi_j(\vec{x})dG(\vec{x}).$ We have
 \begin{align*}
\alpha_j^\theta=\sum_{i\geq 1}\beta_{i,j}\phi_i(\theta)\mbox{ and } \int\left(\alpha_j^\theta \right)^2dF(\theta)=\sum_{i\geq1}\beta^2_{i,j}.
\end{align*}
\end{Lemma}
\begin{proof}
 It follows from projecting $\alpha_j^\theta$ into the basis $\phi$.
\end{proof}
\vspace{2mm}

Similarly, we have the following. 

\begin{Lemma}
 \label{alphaExpansionX}
For each $\vec{x}\in \mathcal{X},$ expand $h_\vec{x}(\theta)$ into the basis $\phi:$ $h_\vec{x}(\theta)=\sum_{i\geq1}\alpha_i^\vec{x}\phi_i(\theta),$
where $\alpha_i^\vec{x}=\int h_\vec{x}(\theta)\phi_i(\theta)dF(\theta).$ We have
 \begin{align*}
\alpha_i^\vec{x}=\sum_{j\geq 1}\beta_{i,j}\psi_i(\vec{x})\mbox{ and } \int\left(\alpha_i^\vec{x} \right)^2dG(\vec{x})=\sum_{j\geq1}\beta^2_{i,j}.
\end{align*}
\end{Lemma}
\vspace{2mm}

\begin{Lemma}
 \label{alphaBeta} Using the same notation as Lemmas \ref{alphaExpansionZ} and \ref{alphaExpansionX}, we have
 $$ \beta_{i,j}=\int \alpha_i^\vec{x}\psi_j(\vec{x})dG(\vec{x})=\int \alpha_j^\theta\phi_i(\theta)dF(\theta).$$
\end{Lemma}
\begin{proof}
 Follows from plugging the definitions of $\alpha_i^\vec{x}$ and $\alpha_j^\theta$ into the expressions above and recalling the definition of $\beta_{i,j}$.
\end{proof}
\vspace{2mm}

\begin{Lemma} 
\label{sobolev}
$\sum_{j\geq J}\int\left(\alpha_j^\theta \right)^2dF(\theta) = c_{K_\x} O(\lambda^\x_J)$ and $\sum_{i\geq I}\int\left(\alpha_i^\x \right)^2dG(\x) = c_{K_\theta} O(\lambda^\theta_I)$.

\end{Lemma}
\begin{proof}
Note that $||h_\theta(.)||^2_{\mathcal{H}_{K_\x}} = \sum_{j \geq 1} \frac{\left(\alpha^\theta_j\right)^2}{\lambda^\x_j}$ .
Using Assumption \ref{assump-hilbertX} and that  the eigenvalues are decreasing it follows that
\begin{align*}
\sum_{j \geq J} \left(\alpha^\theta_j\right)^2=  \sum_{j \geq J} \left(\alpha^\theta_j\right)^2 \frac{\lambda^\x_j}{\lambda^\x_j} \leq \lambda^\x_J ||h_\theta(.)||^2_{\mathcal{H}_{K_\x}} \leq \lambda^\x_J c^2_\theta,
\end{align*}
and therefore $\sum_{j\geq J}\int\left(\alpha_j^\theta \right)^2dF(\theta) \leq  \lambda^\x_J  \int_z c_\theta^2 dF(\theta)= c_{K_\x} O(\lambda^\x_J)$. The proof of the second statement
is analogous to this.

\end{proof}
\vspace{2mm}

\begin{thm}
\label{thm-bias2}
Under the stated assumptions, the bias is bounded by
 $$B(\L_{I,J},\L) = c_{K_\x}O\left(\lambda^\x_J\right) +  c_{K_\theta}O (\lambda^\theta_I).$$
 
\end{thm}
\begin{proof}
By using orthogonality, we have that
\begin{align*}
 B(\L_{I,J},\L) &\overset{\mbox{\tiny{def}}}{=} \iint \left(\L(\vec{x};\theta)-\L_{I,J}(\vec{x},\theta) \right)^2 dG(\vec{x})dF(\theta) \leq \sum_{j>J}\sum_{i \geq 1}\beta^2_{i,j}+\sum_{i>I}\sum_{j \geq 1}\beta^2_{i,j} \\
 &= \sum_{j\geq J}\int\left(\alpha_j^\theta \right)^2dF(\theta)+\sum_{i\geq  I}\int\left(\alpha_i^\vec{x} \right)^2dG(\vec{x}),
\end{align*}
where the last equality follows from Lemmas \ref{alphaExpansionZ} and \ref{alphaExpansionX}. The Theorem follows from Lemma \ref{sobolev} .
\end{proof}
\vspace{2mm}

}

\end{document}